\documentclass[a4paper]{article}
\usepackage[margin=3cm]{geometry}
\usepackage{graphicx}
\usepackage{amssymb, amsmath, amsthm}
\usepackage{bbm, multicol, multirow, hhline}
\usepackage{stmaryrd}
\usepackage{comment}
\usepackage{booktabs}
\usepackage{hyperref}
\usepackage[round]{natbib}

\usepackage[symbol]{footmisc}

\usepackage{tikz}
\usetikzlibrary{arrows,shapes.arrows,shapes.geometric}

\usepackage{amsmath, amsfonts, amsthm, amssymb, bbm, bm}
\usepackage{mathrsfs}
\usepackage{enumerate}
\usepackage{tabularx}
\usepackage{multicol}
\usepackage{multirow}

\newcommand{\E}{\mathbb{E}}
\DeclareMathOperator{\Var}{Var}
\DeclareMathOperator{\Cov}{Cov}

\newcommand{\Ind}{\mathbbm{1}}

\newcommand{\reals}{\mathbb{R}}

\DeclareMathOperator{\expit}{expit}
\DeclareMathOperator{\logit}{logit}

\newcommand{\X}{\mathfrak{X}}

\newcommand\indep{\protect\mathpalette{\protect\independenT}{\perp}}
\def\independenT#1#2{\mathrel{\rlap{$#1#2$}\mkern2mu{#1#2}}}

\theoremstyle{plain}
\newtheorem{lem}{Lemma}[section]
\newtheorem{thm}[lem]{Theorem}
\newtheorem{prop}[lem]{Proposition}
\newtheorem{cor}[lem]{Corollary}

\theoremstyle{definition}
\newtheorem{dfn}[lem]{Definition}
\newtheorem{rmk}[lem]{Remark}
\newtheorem{exm}[lem]{Example}

\newcommand{\benum}{\begin{enumerate}}
\newcommand{\eenum}{\end{enumerate}}

\newcommand{\bitem}{\begin{itemize}}
\newcommand{\eitem}{\end{itemize}}

\newcommand{\barr}{\begin{array}}
\newcommand{\earr}{\end{array}}

\newcommand{\bmat}{\begin{pmatrix}}
\newcommand{\emat}{\end{pmatrix}}

\newcommand{\blist}{\renewcommand{\labelenumi}{\textbf{\arabic{enumi}}.} \begin{enumerate}}
\newcommand{\elist}{\end{enumerate} \renewcommand{\labelenumi}{\arabic{enumi}.}}

\newcommand{\bs}{\boldsymbol}

\def\bal#1\eal{\begin{align*}#1\end{align*}}

\newcounter{exmrun}
\newenvironment{exmp}{\refstepcounter{exmrun}\noindent
   \textbf{Example~R\theexmrun.} \rmfamily}{}

\usepackage{authblk}

\title{Parameterizing and Simulating from Causal Models}
\author{Robin J.\ Evans\thanks{Department of Statistics, University of Oxford, UK}~ and Vanessa Didelez\thanks{Leibniz Institute for Prevention Research and Epidemiology - BIPS and Faculty of Mathematics \& Computer Science, University of Bremen, Germany}}

\DeclareMathOperator{\Cor}{Cor}
\newcommand{\Do}{\mathit{do}}
\renewcommand{\Ind}{\mathbb{I}}

\newcommand{\M}{\mathcal{M}}

\renewcommand{\X}{\mathcal{X}}
\newcommand{\Y}{\mathcal{Y}}
\newcommand{\Z}{\mathcal{Z}}
\newcommand{\cmid}{\,|\,}

\newcommand{\Xt}{A}
\newcommand{\xt}{a}
\newcommand{\Zc}{L}
\newcommand{\zc}{\ell}

\newcommand{\YZIX}{{Y\hspace{-1pt}Z \hspace{-.5pt} | \hspace{-.5pt}X}}
\newcommand{\YZIXC}{{Y\hspace{-1pt}Z \hspace{-.5pt} | \hspace{-.5pt}X\hspace{-.5pt}C}}
\newcommand{\YIXZ}{{Y \hspace{-.5pt} | \hspace{-.5pt}X \hspace{-1pt}Z}}
\newcommand{\YIZX}{{Y \hspace{-.5pt} | \hspace{-.5pt}Z \hspace{-1pt}X}}
\newcommand{\ZXY}{{\hspace{-.5pt}Z \hspace{-1pt} X \hspace{-.75pt} Y}}

\newcommand{\ZX}{{\hspace{-.5pt} Z \hspace{-1pt} X}}
\newcommand{\YIZ}{Y \hspace{-.5pt} | \hspace{-.75pt} Z}
\newcommand{\YZ}{Y \hspace{-1pt} Z}
\newcommand{\YIX}{{Y\hspace{-.5pt}|\hspace{-.5pt} X}}
\newcommand{\ZIX}{{Z \hspace{-.5pt}|\hspace{-.5pt} X}}
\newcommand{\XIZ}{{\hspace{-.75pt} X \hspace{-.5pt}|\hspace{-.5pt} Z}}
\newcommand{\ALB}{{\hspace{-1pt}A\hspace{-.5pt}L\hspace{-.5pt}B}}
\newcommand{\ALBY}{{\hspace{-1pt}A\hspace{-.5pt}L\hspace{-.5pt}BY}}
\newcommand{\YIALB}{{Y\hspace{-.65pt}|\hspace{-.5pt}A\hspace{-.5pt}L\hspace{-.5pt}B}}
\newcommand{\YIAB}{{Y\hspace{-.65pt}|\hspace{-.5pt}A\hspace{-.75pt}B}}
\newcommand{\pZXY}{p_{\ZXY}}

\newcommand{\pYZX}{p_{\YZIX}}
\newcommand{\pX}{p_{\hspace{-1pt}X\hspace{-.5pt}}}
\newcommand{\pZ}{p_{\hspace{-.75pt}Z\hspace{-.5pt}}}

\newcommand{\pYIXZ}{p_{\YIXZ}}
\newcommand{\pYIZX}{p_{\YIZX}}
\newcommand{\pYIZ}{p_{\YIZ}}
\newcommand{\pZX}{p_{\ZX}}
\newcommand{\pZIX}{p_{\ZIX}}
\newcommand{\pYX}{p_{\YIX}}
\newcommand{\phiYZIX}{\phi_{\YZIX}}
\newcommand{\phiYZ}{\phi_{\YZ}}

\newcommand{\thetaZX}{\theta_{\ZX}}
\newcommand{\thetaYX}{\theta_{\YIX}}

\newcommand{\KL}{\mathrm{KL}}

\begin{document}


\maketitle

\addtocounter{footnote}{2}

\begin{abstract}
Many statistical problems in causal inference involve a probability
distribution other than the one from which
data are actually observed; as an additional complication, the
object of interest is often a marginal quantity of this other probability
distribution.  This creates many practical complications for
statistical inference, even where the problem is non-parametrically
identified.
%
In particular, it is difficult to perform likelihood-based inference,
or even to simulate from the model in a general way.

We introduce the `frugal parameterization', which places the causal
effect of interest at its centre, and then builds the rest of the model
around it.  We do this in a way that
provides a recipe for constructing a regular, non-redundant parameterization
using causal quantities of interest.
In the case of discrete variables we can use odds ratios to complete
the parameterization, while in the continuous case copulas are the natural choice;
other possibilities are also discussed.

Our methods allow us to construct and simulate from models with parametrically
specified causal distributions, and fit them using
likelihood-based methods, including fully Bayesian approaches.  
Our proposal includes parameterizations for the average causal effect and effect of  treatment on the treated, as well as other causal quantities of interest.
%
\end{abstract}

%

\section{Introduction}

In many multivariate statistical problems, inferential interest lies in
properties of specific functionals of the joint distribution, such as
marginal or conditional distributions; this means it is generally
desirable to specify a model for these functionals directly, with other
parts of the distribution often being regarded as nuisance parameters.  In
\emph{causal} inference problems, the target of inference may be a
probability distribution other than the one that generates the observed
data, but one which corresponds to some sort of experimental
intervention on that system.

\begin{exm} \label{exm:xyz}
Consider the causal system
represented by the graph in Figure \ref{fig:mod}(a), and suppose we are
interested in the causal effect of $X$ on $Y$.  
For example, in a cohort of children $X$ might be a measure of their diet, 
$Y$ their BMI, and $Z$
an indicator of the education level of their parents.  
Alternatively, $Z$ could be an unobserved genetic factor.


This can be
formulated as a prediction problem: ``what would happen if we
performed an experiment in which we set $X = x$
by external intervention?''
Let the variables be distributed according to $P$
with some density $p$.  Under
the causal DAG assumptions of \citet{spirtes00} and \citet{pearl:09}, the conditional
distribution of $Y$ and $Z$ after an experiment to fix $X = x$ is
\begin{align*}
P^*(Z=z, Y=y \mid X=x) \equiv P(Z=z) \cdot P(Y = y \mid Z=z, X=x).
\end{align*}
Note that the idealized intervention on $X$ removes any
dependence of $X$ on the confounder $Z$, but preserves the
marginal distribution of $Z$, and the conditional distribution of
$Y$ given $X, Z$.  This distribution is Markov with respect to the
graph in Figure \ref{fig:mod}(b). %
Interest may then lie in the marginal effect on just $Y$,
\begin{align}
P^*(Y=y \mid X=x) = \sum_z P(Z=z) \cdot P(Y = y \mid Z=z, X=x), \label{eqn:qp}
\end{align}
sometimes denoted $P(Y=y \mid \Do(X=x))$, or as the distribution of the
\emph{potential outcome} $Y_x$.
Models of this quantity are known as \emph{marginal structural models}
\citep[or MSMs,][]{robins:00}.

\begin{figure}
  \begin{center}
  \begin{tikzpicture}
  [rv/.style={circle, draw, thick, minimum size=7mm, inner sep=0.75mm}, node distance=20mm, >=stealth]
  \pgfsetarrows{latex-latex};
\begin{scope}
  \node[rv]  (1)              {$Z$};
  \node[rv, below of=1, yshift=7.5mm, xshift=-10mm] (2) {$X$};
  \node[rv, right of=2] (3) {$Y$};
\node[left of=1, yshift=-5mm] {(a)};

  \draw[->, very thick, color=blue] (1) -- (3);
  \draw[->, very thick, color=blue] (1) -- (2);
  \draw[->, very thick, color=blue] (2) -- (3);
  \end{scope}
  \begin{scope}[xshift=7cm]
  \node[rv]  (1)              {$Z$};
  \node[rv, rectangle, below of=1, yshift=7.5mm, xshift=-10mm] (2) {$X$};
  \node[rv, right of=2] (3) {$Y$};
  \node[left of=1, yshift=-5mm] {(b)};

  \draw[->, very thick, color=blue] (1) -- (3);
  \draw[->, very thick, color=blue] (2) -- (3);
  \end{scope}

    \end{tikzpicture}
 \caption{(a) A causal model with three variables; (b) the same model after intervening on $X$.}
  \label{fig:mod}
  \end{center}

\end{figure}
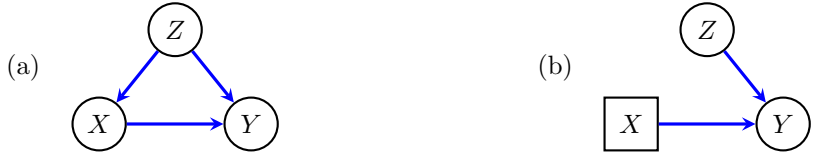

For the purposes of simulation and likelihood-based inference it is often
necessary to work with the joint distribution $P(X=x,Y=y,Z=z)$ directly, and
it may be difficult to specify it so that it remains compatible with a particular
marginal model on (\ref{eqn:qp}).  Indeed, providing a model for the joint distribution
parametrically may lead to a situation in which (\ref{eqn:qp}) cannot logically
be independent of the value of $x$, unless we impose the much stronger condition
that $Y \indep X \mid Z$.
More generally, specifying separate models for joint and marginal
quantities---and ignoring the information that is shared between
them---can lead to incompatible or incoherent
models, non-regular estimators, and severe misspecification problems.
\end{exm}

\subsection{Contribution of this Paper}

We will show that one can break down a joint distribution
into three pieces: the distribution of `the past',
$\pZX(z,x) := P(Z=z, X=x)$; the causal quantity of interest,
$\pYX^*(y \cmid x) := P^*(Y=y \cmid X=x)$; and a conditional
odds ratio, copula or other dependence measure $\phiYZIX^*$
between $Y$ and $Z$ given $X$.  
Suppose that the respective parameterizations
for these quantities are called $\thetaZX$, $\thetaYX^*$ and
(with some abuse of notation) $\phiYZIX^*$; we call
$(\thetaZX, \thetaYX^*, \phiYZIX^*)$ a
\emph{frugal parameterization}.  The terminology is chosen because
it is a direct parameterization of the causal quantity of interest, 
such that there is no redundancy and \emph{any} distribution with
a positive joint density can be decomposed in this manner. 
If we use smooth and regular\footnote{That is, such that the model is differentiable
in quadratic mean and has positive definite Fisher Information Matrix.
See Appendix \ref{sec:reg} for a formal statement.}
parameterizations of the three pieces then the resulting
parameterization of the joint model is also smooth and regular.
We use a star (e.g.\ $p^*$ or $\phi^*$)
to denote that the distribution or parameter is from the
\emph{causal} or \emph{interventional} distribution, and omit
the star if the distribution or parameter is from the
\emph{observational} regime.
Note that the causal quantity $\pYX^*$ 
may be more general than just $\pYX(y \cmid \Do(x))$; see
Section \ref{sec:prereq}.  

%

Note that, in addition to providing a parameterization, the quantities
$\thetaZX$ and $\thetaYX^*$ will always be \emph{variation independent};
we can also always choose $\phiYZIX^*$ to be variation independent of the
other two parameters, unless we prefer to use (e.g.) a risk difference or 
risk ratio for interpretability.
As an example of the benefits of this property, 
we add a dependence for $Y$ on covariates $C$ via a link function:
\begin{align*}
  \logit P^*(Y=1 \cmid X=x, C=c) &= \mu + \alpha x +  \bs \beta c +   \bs \gamma x c, && \text{for all } c.
\end{align*}
Now we can be certain that---regardless of the values of $P(X=x,Z=z,C=c)$ and
$\phi_{\YZIXC}^*(y, z \cmid x,c)$---there is a coherent joint distribution
which possesses the required functionals.  This could allow us, for example, to
model the causal effect of alcohol ($X$) on blood pressure ($Y$) conditional on
a person's genes ($C$), but marginally over factors such as socio-economic
status ($Z$).

We start with a very simple example, to illustrate exactly what we
propose to do.

\begin{exm} \label{exm:gauss}
Suppose that $(Z,X,Y)^T$ follow a multivariate Gaussian distribution
with zero mean, and that we wish to specify that $Y \cmid \Do(X=x)$
is normal with mean $\beta x$ and variance $\sigma^2$.
To complete the frugal parameterization
we must specify `the past' (i.e.\ $\pZX$) and a dependence measure
between $Y$ and $Z$ conditional upon $X$ ($\phiYZIX^*$).
We therefore take $Z$ and $X$ to be normal with mean 0 and variances
$\tau^2,\upsilon^2$ respectively and correlation $\rho$,
and assume the regression parameter for $Y$ on $Z$ (in the regression that includes $X$) is $\alpha$;
note that we could alternatively specify the covariance or partial
correlation between $Z$ and $Y$.  Hence we have $\theta_{ZX} = (\tau^2, \upsilon^2, \rho)$, $\thetaYX = (\beta, \sigma^2)$ and $\phiYZIX^* = \alpha$.
Using this information, one can directly compute the distribution of
$(Z,Y)^T$ after the intervention
\begin{align*}
\left( \begin{matrix} Z \\ Y\end{matrix} \middle) \right| \Do(X=x) \sim N_2 \left(
\left(\begin{matrix} 0 \\ \beta x \end{matrix}\right), \;
\left(
\begin{matrix} \tau^2 & \alpha \tau^2 \\
\alpha \tau^2 &  \sigma^2
\end{matrix}
\right) \right),
\end{align*}
and consequently the observational joint distribution of $(Z,X,Y)^T$ is:
\begin{align*}
\left( \begin{matrix} Z \\ X \\ Y\end{matrix} \right) \sim N_3 \left(0, \; \left(
\begin{matrix} \tau^2 & \rho\tau\upsilon & \alpha \tau^2 + \beta \rho \tau \upsilon \\
\rho\tau\upsilon & \upsilon^2 & \beta \upsilon^2 + \alpha \rho \tau \upsilon\\
\alpha \tau^2 + \beta \rho \tau \upsilon & \beta \upsilon^2 + \alpha \rho \tau \upsilon &  \sigma^2 + \beta^2\upsilon^2 + 2\rho\tau\upsilon\alpha\beta
\end{matrix}
\right) \right).
\end{align*}
%
We may do this for any value of $\rho \in (-1,1), \alpha, \beta$, and
$\sigma^2,\tau^2,\upsilon^2 > 0$ provided that $\sigma^2 > \alpha^2 \tau^2$,
and indeed we can obtain
\emph{any} trivariate Gaussian distribution from these parameters.  Note 
that, though the last inequality implies there is variation dependence 
in this case, we could easily 
choose $\phiYZIX^*$ to be (for example) the partial correlation between 
$Z$ and $Y$ given $X$, and then there would be no such constraint.

%

Once we are able to construct the joint distribution, simulation is trivial.  We
take the Cholesky decomposition of the covariance matrix and apply the lower triangular part
to independent standard normals.  Likelihood-based inference is also straightforward
once the covariance is known.

In this example we took our three pieces, $\pZX$ (a bivariate normal),
$\pYX^*$ (a linear regression) and
$\phiYZIX^*$ (a regression parameter), and used them to obtain $\pZXY$.
Note that our parameterization was
chosen so that every quantity of interest is specified precisely once,
and the overall model is saturated (i.e.\ any multivariate Gaussian
distribution can be
deconstructed in this manner, just by varying the parameters).
This contrasts with the alternative of  specifying $\Sigma$ directly,
as this does not give a simple explicit model for the causal effect.
\end{exm}

The above example may seem somewhat trivial, but the main contribution
of this paper is that we will do this in a much more general fashion, enabling
simulation from a wide range of causal models.

\begin{exm} \label{exm:binYZ}
Now suppose that $Z$ and $Y$ are binary with $X$ still continuous, and
we continue to work with the model in Figure \ref{fig:mod}(a).  This time we
specify that
\begin{align*}
\logit \E[Y \cmid \Do(X=x)] = \beta_0 + \beta_1 x;
\end{align*}
in addition suppose $\E Z = q$, that $X \mid Z=z \sim N(\gamma z, \sigma^2)$,
and that the log odds ratio between $Y$ and $Z$ conditional on $X=x$ is
$\phi$ (we could also allow $\phi$ to vary with $x$).

The joint distribution in this example is considerably more difficult
to write in a closed form than the one in Example \ref{exm:gauss}.
However, in this paper we will show that we may:
(i) specify this model using the parameters just given;
(ii) simulate samples from the distribution described; and
(iii) give a map to numerically evaluate the joint density and fit such
a model to data using likelihood-based methods.  Furthermore,
we can do all this (almost) as easily as with the multivariate
Gaussian distribution.  Note that because logistic regression is
not collapsible, this model illustrates why we should not just provide
$p_{\YIXZ}$ to compute the joint distribution: doing so could lead
to a very different marginal model for $Y \cmid \Do(X)$ than the one we
chose.
\end{exm}

As we show, the method is particularly applicable to survival models
and dynamic treatment models where we marginalize over the time-varying
confounders; both of these are widely used but are
difficult to simulate from \citep{havercroft12, young:14}.  In addition,
it allows Bayesian and other likelihood-based methods to be applied
coherently to marginal causal models \citep{Saarela:15}.
%

Though Examples \ref{exm:xyz}--\ref{exm:binYZ} are presented for
univariate Gaussian or discrete variables, in fact
the results are entirely general and can be adapted to vectors of arbitrary
cardinality and general continuous or mixed variables; implementation does
become more complicated in such situations, however.  As noted by \citet{robins:00},
calculation of the likelihood becomes a `computational nightmare' for
marginal structural models with continuous variables, but we show that
copulas can be used to overcome this problem.  In the
sequel we denote the observational joint density by $p$ with, for example,
$\pYX(y \cmid x)$ meaning the conditional density of $Y$ given $X$.  In
the discrete case, this is just the probability mass function.

\subsection{Existing Work} \label{sec:existing}

A commonly used alternative to likelihood-based approaches are
generalized estimating equations (GEEs) or semiparametric methods,
as these do not require full
specification of the joint distribution \citep{diggle02}.  However,
neither method allows for simulation from the model, and they may be
less powerful than likelihood-based methods.

\citet{robins:92} provides an algorithm for simulating from a
\emph{Structural Nested Accelerated Failure Time Model} (SNAFTM), a survival model in
which one models survival time as an exponential variable whose
parameter varies with treatment.  This is adapted by \citet{young:08}
to simulate from a Cox MSM model.  \citet{young:09} consider a
special case of a Cox MSM that approximates a SNAFTM and also a SNCFTM
(special cases of the \emph{structural nested model}---see Section \ref{sec:snm}).  
\citet{keogh21simulating} give a method for simulating from Cox MSMs using 
an additive hazard model.
\citet{havercroft12} consider the problem of specifying (and thus characterizing) models such that, for simulation and educational purposes, bias due to selection effects and blocking mediation effects will be strong if a na\"ive approach is used.

\citet{richardson:17:oddsproduct} give a variation independent
parameterization for structural equation models by using the odds product;
this also allows for fully-likelihood based methods.
This is extended by \citet{wang:22} to the
Structural Nested Mean Model, which we will meet in Section \ref{sec:snm}.
The main difference between this work and ours is that it is not obvious how
to extend their approach to other models and to continuous variables.

Indeed, much of the trend in causal inference is towards structural equations
models (SEMs) in which each variable is modelled as a function of all previous variables
and a stochastic noise term \citep{peters:17}.  In Example \ref{exm:binYZ} this would
have meant specifying $p_{\YIZX}$, which would not have allowed us to
directly model $\pYX(y \cmid \Do(x))$.  In particular, the work of
Pearl generally assumes that causal distributions should be
conditional on all previous variables, while allowing for some
conditional independence constraints (see, for example, \citealp{peters:17},
and large sections of \citealp{pearl:09}).  We certainly do not wish to single these authors out for criticism (indeed
the authors of this paper have often considered such approaches), but
they do seem to be less useful in epidemiological or other medical
contexts, in which conditional independences are often---though not
always---implausible assumptions.  In such a context, one has to specify
distributions conditional on the entire past, which may be very difficult
if there are a large number of relevant variables.

We view our approach as complementary to the structural equation perspective,
since each has advantages in terms of what assumptions can be expressed and 
the causal questions that can be easily answered within the framework.  SEMs 
and the theory around them have received much attention; this work starts to
fill in the gaps relating to marginal models. 


\subsection{Causal Models}

%
%

Throughout the paper we will have a running example based on Figure
\ref{fig:mod2}; each of these examples is labelled with a prefix `R'.

\begin{exmp}  \label{exm:run}
The model in Figure \ref{fig:mod2} arises in
dynamic treatment models and is studied in \citet{havercroft12}.   The
variables $A$ and $B$ represent two treatments and so play the role of $X$ 
from Example \ref{exm:xyz}; 
the second treatment $B$ depends on both the first ($A$) and an intermediate outcome $L$.  The
variable $U$ is `hidden' or latent, and therefore identifiable quantities
are functions of $p_{\ALBY}$.  A typical quantity of interest is
the distribution of the outcome $Y$ after interventions on
the two treatments $A$ and $B$.  Under the assumption of positivity and the 
causal structure implied by the graph, this is identified by the g-formula
of \citet{robins86} as
\begin{align}
p_{\YIAB}(y \cmid \Do(a,b)) &:= \int p_{\YIALB}(y \cmid a,\ell,b) \cdot p_{L|A}(\ell \cmid a) \, d\ell. \label{eqn:do}
\end{align}
Havercroft and Didelez
note that after specifying a model for $p_{\YIAB}(y \cmid \Do(a, b))$, it
is difficult to parameterize and simulate from the full joint
distribution, partly because of the complexity of the relationship (\ref{eqn:do}).
They are only able to simulate from the special case of Figure \ref{fig:mod2}
in which $L$ has no direct effect on $Y$, so any dependence
is entirely due to the latent variable.  We remark that we could replace instances 
of $\ell$ in
(\ref{eqn:do}) with $(u, \ell)$ and obtain the same result, which means that the
role of $Z$ could be taken by either $L$ alone or the pair $(U,L)$.
\end{exmp}

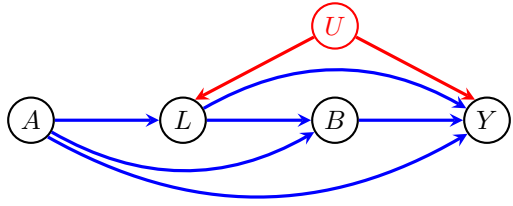
\begin{figure}
  \begin{center}
  \begin{tikzpicture}
  [rv/.style={circle, draw, thick, minimum size=6mm, inner sep=0.75mm}, node distance=20mm, >=stealth]
  \pgfsetarrows{latex-latex};
\begin{scope}
  \node[rv]  (1)              {$A$};
  \node[rv, right of=1] (2) {$L$};
  \node[rv, right of=2] (3) {$B$};
  \node[rv, right of=3] (4) {$Y$};
  \node[rv, above of=3, color=red, yshift=-7.5mm] (U) {$U$};


\draw[->, very thick, color=blue] (1) -- (2);
  \draw[->, very thick, color=blue] (2) -- (3);
  \draw[->, very thick, color=blue] (3) -- (4);
\draw[->, very thick, color=blue] (1.315) to[bend right] (4.215);
  \draw[->, very thick, color=blue] (1) to[bend right] (3);
  \draw[->, very thick, color=blue] (2) to[bend left] (4);
  \draw[->, very thick, color=red] (U) -- (2.60);
  \draw[->, very thick, color=red] (U) -- (4.120);

  \end{scope}
    \end{tikzpicture}
 \caption{The causal model from \citet{havercroft12}.
}
  \label{fig:mod2}
  \end{center}
\end{figure}

For related reasons, the model in Figure \ref{fig:mod2} is also the subject of the so-called
\emph{g-null paradox} \citep{robins:97} when testing the hypothesis of whether $p_{\YIAB}(y \cmid \Do(a,b))$
depends upon $A$.
This arises because seemingly innocuous parameterizations of
the conditional distributions $p_{\YIALB}(y \cmid a, \ell, b)$ and 
$p_{L|A}(\ell \cmid a)$ (e.g.~a linear and a logistic regression) 
lead to situations where the null hypothesis can almost never hold: that is, 
it is impossible for $p_{\YIAB}(y \cmid \Do(a,b))$ not to depend upon $A$
unless either $L$ or $Y$ is completely independent of $A$.  The reason for the `paradox'
can be understood as a problem of attempting to specify the relationship between
$Y$ and $A$ in two different and potentially incompatible ways.

Note that the g-null paradox is not the same as the presence of singularities\footnote{See Appendix \ref{sec:reg} for a formal definition.} or non-collapsibility,
but rather it is a \emph{result} of non-collapsibility over a marginal model that
possibly \emph{leads to} singularities.


\begin{exmp} \label{exm:gnull}
Considering Figure \ref{fig:mod2} again, suppose that we choose $Y$ to depend linearly
on $A$, $L$ and $B$ (including any interactions we wish), and that
$L$ is binary and we use a logistic parameterization
for its dependence upon $A$.  Then, if $A$ takes four or more distinct values, it is
essentially impossible for $H_0: Y \indep A \mid \Do(B)$ to hold in such a distribution,
even if $Y$ does not depend directly upon $A$, $L$ or $B$.  This is because
\begin{align*}
\E [Y\cmid \Do(a,b)] &= \sum_{\ell=0}^1 p_{L|A}(\ell \mid a) \cdot \E [Y \mid a, \ell, b]\\
&= \beta_0 + \beta_1 a + \beta_3 b + \expit(\theta_0 + \theta_1 a) \beta_2,
\end{align*}
so the only way for this quantity to be independent of a variable $A$ with at
least four levels is for $\beta_1 = 0$ and either $\theta_1=0$ or
$\beta_2=0$.  This `union' model is singular (i.e.~not regular) at
$\theta_1 = \beta_1 = \beta_2 = 0$, and being in it implies a much
stronger null hypothesis (that either $Y \indep A, L \mid B$ or $L \indep A$
in addition to the causal independence) than the one we
are interested in investigating.
\end{exmp}

Generally speaking, if we try to state a model for
$p_{\YIALB}$ as well as requiring that $p_{\YIAB}(y \cmid \Do(a,b))$ does
not depend on $A$, we effectively try to specify the $A$-$Y$ and $B$-$Y$ 
relationships in
two different margins; in the case above these margins are incompatible,
leading to the singularity.
This is avoided by constructing a smooth, regular and variation independent
parameterization, without any redundancy.
We show that, in fact, a frugal parameterization of the joint distribution
exists that separates into variation independent parameterizations of the
quantities
\begin{align*}
p_{\ALB}(a,\ell,b), && p_{\YIAB}(y \cmid \Do(a, b)), && \text{and }\quad \phi^*_{L\YIAB}(\ell,y \cmid a,b).
\end{align*}
This entirely avoids the g-null paradox when considering
hypotheses about $p_{\YIAB}(y \cmid \Do(a, b))$, since variation
independence means that it may be freely
specified.  In addition this parameterization is such that 
one can logically specify any distribution with a joint density in this manner.

Note that the example above does not give a separate specification of the
dependence of $Y$ on $L$ that is causal, and the spurious dependence
due to the latent parent $U$: both kinds of dependence are
tied up in the association parameter $\phi^*_{L\YIAB}$.  An
alternative is to explicitly include $U$ in the model, leaving us
with
\begin{align*}
p_{U\!ALB}(u,a,\ell,b), && p_{\YIAB}(y \cmid \Do(a, b)), &&\text{and }\quad \phi^*_{U\!L,Y\hspace{-.5pt}|\hspace{-.5pt}AB}(u,\ell; y \cmid a,b),
\end{align*}
where $\phi^*_{U\!L,\YIAB}$ has to model the dependence between $Y$ and
$(U,L)$, after intervention on $A,B$.  Of course, some of these quantities
will be unidentifiable,\footnote{Specifically, $p_{U\hspace{-.5pt}|\hspace{-.5pt}ALB}(u \cmid a,\ell,b)$ and $\phi^*_{U\!L,Y\hspace{-.5pt}|\hspace{-.5pt}AB}(u,\ell ; y \cmid a,b)$.}
but we will want to be able to simulate how well the effects of $A$ and $B$
on $Y$ are estimated in the presence of unobserved confounding of
various strengths.

\begin{rmk} \label{rmk:notation}
Statistical causality is represented using a number of different overlapping
frameworks, including potential outcomes \citep{rubin74}, causal directed graphs
\citep[e.g.][]{spirtes00}, decision theory \citep{dawid_did:10}, non-parametric
structural equation models \citep[e.g.][]{pearl:09}, Finest Fully Randomized
Causally Interpretable Structured Tree Graphs \citep{robins86} and their implementation
as Single World Intervention Graphs \citep{richardson13}.  The discussions in this paper
are broadly applicable to any of these frameworks.
%
For notational purposes
we choose to use Pearl's `$\Do(\cdot)$' operator to indicate interventions.
For example, $P(Y = y \cmid A = a; \Do(B=b))$ refers to the conditional distribution
of $Y$ given $A=a$ under an experiment where $B$ is fixed by intervention
to the value $b$.  We generally abbreviate this to $p_{\YIAB}(y \cmid a; \Do(b))$.
The same quantity in the potential outcomes
framework would generally be denoted by $P(Y_{b} = y \cmid A_b = a)$.

Though slightly more verbose, the $\Do(\cdot)$ notation has the advantage that
the quantity is more immediately seen to be a conditional distribution
indexed by both $a$ and $b$, which is critical to our method.  
We will exploit the fact that a $do(X=x)$-intervention can be obtained by 
conditioning on $X=x$ after randomizing $X$, i.e.~randomly generating it from 
an arbitrary (but not trivial) distribution $\pX^*(x)$. 
%
Note also that it is ambiguous
from notation alone whether $\pYX(y \cmid \Do(x))$ is identifiable or not, since it
depends upon both the causal model being postulated and the available data;
this problem also arises with the other frameworks.
\end{rmk}

\begin{rmk} \label{rmk:identifassump}
In applications, when causal models are to be fitted on actual data,  conditions for identifiability must be met. These are well-known for all models we consider: they essentially consist of the appropriate (possibly sequential) versions of causal consistency, positivity and conditional exchangeability (or no unmeasured confounding)  given the measured covariates \citep[][]{hernan:20}.  As we are here interested in properties of causal models and how to simulate from them, we will take identification as given. 
\end{rmk}

The remainder of the paper is structured as follows: in Section \ref{sec:prereq} we
provide our main assumptions and discuss issues such as how we might choose a 
dependence measure.  Section \ref{sec:param} contains the main result outlined in 
the introduction.  In Section \ref{sec:sim} we
describe how to simulate from our models and give a series of examples, and in
Section \ref{sec:fitting} we show how to fit these models using maximum likelihood estimation.
Section \ref{sec:data} contains an analysis of real data on the relationship
between fibre intake, a polygenic risk score for obesity and children's BMI.
Section \ref{sec:survival} discusses an application of the frugal parameterization 
to survival models, and Section \ref{sec:snm} contains an extension
to models in which the causal parameter is different for distinct levels
of the treatment. 
We note that Sections \ref{sec:survival} and \ref{sec:snm} are more technical, 
and not necessary for the reader to gain insight into the main ideas of the 
paper. We conclude with a discussion in Section \ref{sec:con}.


%

\section{The Frugal Parameterization} \label{sec:prereq}

Here we present a formalization of the ideas in the introduction.
Suppose we have three random vectors $(Z,X,Y) \in \mathcal{Z} \times \mathcal{X} \times \mathcal{Y}$,
where $Y$ is an outcome (or set of outcomes) of interest, and $X,Z$ consist of relevant variables
that are considered to be causally prior to $Y$; this may be 
because they are temporally prior to $Y$, but that is not strictly necessary.
There is no restriction on the state-spaces of these
variables provided that they admit a joint density $p := \pZXY$
with respect to a product measure $\mu_Z \cdot \mu_X \cdot \mu_Y$,
and satisfy standard statistical regularity conditions.  In particular,
each of $X$, $Y$ and $Z$ may be finite-dimensional vector valued, and either
continuous, discrete or a mixture of the two. 
The fact that each of these variables may be vector valued, and that there
is no fixed ordering on variables in $X$ and $Z$ means
the method is considerably more flexible than it might at first appear.

Throughout this paper we use the notation $\pX$ to denote the
marginal density of the random variable $X$, and $\theta_{X}$ to
denote the parameter in a model for this distribution; similarly $\pYX$
and $\thetaYX$ relate to the distribution of $Y$ conditional
upon $X$.  We will need to consider marginal and conditional distributions
that are not obtained by the usual operations; for example, a
marginal distribution taken by averaging over a population with
a different distribution of covariates. 
We will typically denote such
non-standard distributions by indexing with a star: e.g.\ $\pYX^*$ or
$\thetaYX^*$.

We use $\phiYZ$ to denote parameters that describe
the dependence structure of a joint distribution; specifically,
such that when combined with the relevant marginal distributions
they allow us to recover an entire joint distribution.  Examples
include odds ratios or the parameters of a particular copula.
We also consider quantities that provide such a dependence
structure conditional on a third variable, and denote this as
$\phiYZIX$.  Again, if the dependence is in $\pZXY^*$ (defined in the
next subsection) then we will write this quantity as $\phiYZIX^*$.

We will assume that we have three separate, smooth and regular parametric models
for $\pZX$, $\pYX$ and $\phi_\YZIX$, with
corresponding parameters $\thetaZX$, $\theta_\YIX$ and $\phi_\YZIX$.
In this sense our model can be equated with
$\theta := (\thetaZX, \theta_\YIX, \phi_\YZIX)$, and for this reason
we will often refer to $\theta$ as `the model'.
For convenience, we will refer to $\pZXY$ as the \emph{observational}
distribution, and $\pZXY^*$ as the \emph{causal} distribution; we do
this even though in other possible contexts $\pZXY^*$ might not correspond 
to a standard causal intervention on $\pZXY$.

%

\subsection{Cognate Distributions and the Frugal Parameterization}

A parameterization is said to be frugal if it consists of at
least three parts: the distribution of `the past'; a (possibly) reweighted 
quantity relating to the distribution of the outcome; and then a
conditional association measure that, combined with the first two 
pieces, smoothly parameterizes the joint distribution.  

To be explicit, we require that the frugal parameterization 
includes a parameter $\thetaYX^*$ that models a conditional
distribution of the form
\begin{align}
\pYX^*(y \cmid x) &= \int_{\Z} \pYIZX(y \cmid z,x) \cdot w(z \cmid x) \, dz, && x \in \X, y \in \Y, \label{eqn:cognate}
\end{align}
for some kernel (i.e.~conditional density) $w(z \cmid x)$.  We 
call a conditional distribution that can be written in the form
(\ref{eqn:cognate}) a \emph{cognate} distribution (to $\pYX$).  
Note that cognate distributions include the ordinary conditional
as a special case, since setting $w = \pZIX$ we obtain
\begin{align*}
\int_{\mathcal{Z}} \pYIZX(y \cmid z,x) \cdot \pZIX(z \cmid x) \, dz = \pYX(y \cmid x).
\end{align*}
Common causal quantities obtained by reweighting also satisfy the definition; 
for example, given the causal model implied by Figure \ref{fig:mod}(a) we have
\begin{align*}
\pYX(y \cmid \Do(x)) \equiv \int_{\mathcal{Z}} \pYIZX(y \cmid z, x) \cdot \pZ(z) \, dz.
\end{align*}
In other words, this formulation allows
for adjustment by a subset of the previous variables.
Terms to derive the \emph{effect of treatment on the treated} (ETT) also
satisfy the definition by using the kernel $w(z) = \pZIX(z \cmid 1)$; 
the ETT considers the difference between $\E [Y \cmid X=1; \Do(X=x)]$ for
$x=1,0$, and these can be written as
\begin{align*}
\E [Y \cmid X=1; \Do(X=x)] &= \iint y \cdot \pYIZX(y \cmid z, x) \cdot \pZIX(z \cmid 1) \, dy \, dz.
\end{align*}
The effect of treatment on the control individuals (ETC) is analogously 
defined using $\pZIX(z \cmid 0)$.

It is straightforward to check that $\pYX^*$ is itself a kernel for $Y$ given $X$.
One may think of $\pYX^*$ as being a conditional distribution taken
from the larger distribution $\pZXY^*$, where
\begin{align*}
\pZXY^*(z,x,y) &= \pZXY(z,x,y) \cdot \frac{\pZX^*(z, x)}{\pZX(z, x)} \\
&= \pZXY(z,x,y) \cdot \frac{\pX^*(x) \cdot w(z \cmid x)}{\pZX(z, x)} \\
	&= \pX^*(x) \cdot w(z \cmid x) \cdot \pYIZX(y \cmid z,x). 
\end{align*}
Note that $\pZXY$ and $\pZXY^*$ share a conditional distribution
for $Y$ given $X,Z$---only the
marginal distribution of $Z$ and $X$ has been altered.  As noted in Remark 
\ref{rmk:notation} the marginal distribution
$\pX^*$ is essentially arbitrary, though later we may need it to satisfy some of
Assumptions A\ref{ass:smoo2}--A\ref{ass:kl} in order to apply our main results.

\begin{dfn}
A smooth, regular parameterization 
of random variables $(Z,X,Y)$ is said to be \emph{frugal} with respect 
to some kernel $\pYX^*$ of the form (\ref{eqn:cognate}), if it consists of 
separate parameterizations of: 
(i) the marginal distribution of $Z,X$; (ii) the kernel $\pYX^*$; 
and (iii) a conditional association measure $\phiYZIX^*$ for $Y$ and $Z$ 
given $X$.
\end{dfn}

Recall that the formal definitions of `smooth' and `regular' parameterizations 
are given in Appendix \ref{sec:reg}.
%

\subsection{Variation Independence} \label{subsec:vi}

Take a set $\Theta$ and two functions defined on it $\phi, \psi$.  We say 
that $\phi$ and $\psi$ are 
\emph{variation independent} if $(\phi \times \psi)(\Theta) = \phi(\Theta) \times \psi(\Theta)$; i.e.\ the range of the pair of functions together is equal to
the Cartesian product of the range of them individually.  
A variation independent
parameterization helps to ensure that the parameters characterize separate,
non-overlapping aspects of the the joint distribution.  Note that we may
sometimes refer to sets of distributions being variation independent, and in
this case we are really referring to their respective parameterizations.

The parameterizations of $\pZX$ and $\pYX^*$ are guaranteed to
be variation independent, since there is always a parameter cut between marginal 
and conditional pieces of this form (we discuss this in Section \ref{sec:fitting}).  
The following assumption will not 
actually be required for any of our results, but we note that, if satisfied, it 
makes interpretation and prediction somewhat easier.

\begin{enumerate}[{A}1.]
\item \label{ass:vi2} 
Given a frugal parameterization $\theta=(\thetaZX, \thetaYX, \phiYZIX)$, 
the parameter $\phi_\YZIX$ is jointly variation
independent of $\thetaZX$ and $\thetaYX$.

\end{enumerate}

We will see that this assumption is satisfied by both conditional odds ratios 
and copulas.  
\subsection{Choices of the Association Parameter}

Now that we have formally defined the parameterization, let us return to 
the original problem.  We want to be able to (i) construct, (ii) simulate from,
and (iii) fit a model using the frugal parameterization.  In 
order to do this we have to make some choices.  We take the form of $w$ and a model 
for $\pYX^*$ as given, because they are chosen by the analyst using subject 
matter considerations; this leaves us to select a parametric family $\pZX$
for $(Z,X)$, and a conditional association parameter within the causal model,
$\phiYZIX^*$.

This raises the question of how one should choose the association parameter.  In
general there are many possibilities: a risk difference or ratio, an odds ratio, 
or something else.  However, some of these objects have nicer properties than 
others.
%
In the case of binary $Y$ and $Z$ the natural choice for such an object is the conditional
odds ratio
\begin{align*}
\phiYZIX^*(x) \equiv \frac{\pYZX^*(1,1 \cmid x) \cdot \pYZX^*(0,0 \cmid x)}{\pYZX^*(1,0 \cmid x) \cdot \pYZX^*(0,1 \cmid x)},
\end{align*}
which is known to be variation independent of the margins $\pYX^*$ and $p_{\ZIX}^*$,
and also has the property that if $\pZXY^*$ is multiplied by any function
of $(x,z)$ or $(x,y)$ it does not change.  %
%
More specifically, note that 
$\pZXY^* = \pZXY \cdot \pZX^*/\pZX$, and hence
\begin{align}
\phiYZIX^*(x) = \frac{\pYZX^*(1,1 \cmid x) \cdot \pYZX^*(0,0 \cmid x)}{\pYZX^*(1,0 \cmid x) \cdot \pYZX^*(0,1 \cmid x)}
= \frac{\pYZX(1,1 \cmid x) \cdot \pYZX(0,0 \cmid x)}{\pYZX(1,0 \cmid x) \cdot \pYZX(0,1 \cmid x)} = \phiYZIX(x). \label{eqn:or_inv}
\end{align}
In other words, the conditional odds ratio for the causal and observational
distributions are the same, and this does not hold for other conditional association
parameters \citep{edwards63association}.
%

This definition and the invariance result (\ref{eqn:or_inv}) extends to 
distributions over any statespace under mild conditions \citep{osius09asymptotic}, 
and---in theory---the joint distribution 
can be recovered from the odds
ratio and marginal distributions using the \emph{iterative
proportional fitting} (IPF) algorithm \citep{bishop:67, darroch:72:ipf, 
csiszar:75, ruschendorf:95}.  
Other fitting approaches are discussed by \citet{tchetgen10}.  
Note that, for general continuous distributions, it is not possible to
implement the algorithm in practice in most cases, because the intermediate
distributions will not have a closed form;
an obvious exception to this is the multivariate
Gaussian distribution. 

Alternative possibilities include the risk difference and risk ratio, though
these lack the variation independence in A\ref{ass:vi2} possessed by the odds ratio,
unless combined with the odds product as in \citet{richardson:17:oddsproduct}.  We
will use these difference and ratio contrasts in Section \ref{sec:snm}, to 
parameterize the `blip' functions in a structural nested mean model.

%
%

\begin{prop} \label{prop:disc}
If $X$, $Y$ and $Z$ are finite categorical variables and have strictly positive 
conditional distribution $\pYZX > 0$,
then using smooth parameterizations of the marginal distributions $\pZX$ and $\pYX$, 
together with the conditional odds 
ratio $\phiYZIX$ is a frugal parameterization that satisfies assumption 
A\ref{ass:vi2}.  Indeed, $X$ can also be a continuous or mixed variable 
(c.f.\ Example \ref{exm:binYZ}).
\end{prop}

\begin{proof}
This follows from the results of \citet{br02}.
\end{proof}

\begin{exm}
For multivariate Gaussian random variables, or other distributions that
are defined by their first two moments, the partial correlation
$\rho_{\YZIX} \equiv \Cor(Y,Z \cmid X)$
satisfies the conditions for being a conditional association parameter
$\phiYZIX$, in the sense that when combined with the marginal
distributions for each of $Y$ and $Z$ given $X$, one can recover the
joint conditional distribution $p_\YZIX$.  
\end{exm}

\begin{exm} \label{exm:cop}
An alternative to the odds ratio for general continuous variables
is to use a \emph{copula}, which separates
out the dependence structure from the margins by rescaling
the variables via their univariate cumulative distribution functions.
A multivariate copula is a cumulative distribution function
with uniform marginals; i.e.\ a function
$C : [0,1]^d \rightarrow [0,1]$
which is increasing and right-continuous in each argument, and
such that $C(1,\ldots,1,u_i,1,\ldots,1) = u_i$ for all $u_i \in [0,1]$ and
$i \in \{1,\ldots, d\}$.

Recall that, for a continuous real-valued random variable
$Y$ with CDF $F_Y$, the random variable $U \equiv F_Y(Y)$
is uniform on $(0,1)$.  The bivariate copula model for $Y$ and $Z \in \mathbb{R}$
is then
\begin{align*}
C_{Y\!Z}(u, v) \equiv P(F_Y(Y) \leq u, F_Z(Z) \leq v), && u,v\in [0,1].
\end{align*}
There is a one-to-one correspondence between 
copulas and multivariate continuous CDFs with uniform marginals.  By Sklar's
Theorem (\citealp{sklar:59}, see also \citealp{sklar:73}), any
copula can be combined with any collection of continuous margins
to give a joint distribution, via (in our bivariate example)
\begin{align*}
F_{Y\!Z}(y, z) \equiv C(F^{-1}_Y(y), F^{-1}_Z(z)), && y,z \in \reals.
\end{align*}
We will assume that the copula is
parametric, and then $\phiYZIX$ represents the parameters of the
particular family of copulas.
\end{exm}

\begin{prop} \label{prop:cont}
If $Y$ and $Z$ are continuous with a positive conditional distribution for each $x \in \X$,
then any smooth and regular parameterization of their marginals
$\pZX$ and $\pYX$ together with a smooth and regular conditional
copula $C_{\YZIX}$ is a frugal parameterization that satisfies assumption A\ref{ass:vi2}.
\end{prop}

\begin{proof}
This follows from the results of \citet{sklar:73}.
\end{proof}

Note that the copula is only used to model the interaction, thus allowing us 
to retain the simple interpretation of the marginal model $\pYX^*$ in terms of 
an interventional distribution. 
In contrast to the odds ratio note that conditional copulas 
do not satisfy (\ref{eqn:or_inv}), 
because the copula also depends upon the cumulative distribution function of
the corresponding margins; this is a slight disadvantage in comparison to the 
odds ratio.  
We will return to these examples in Section \ref{sec:sim}. 


\begin{exm} \label{exm:mix}
We can also use copulas to model variables in a more flexible 
way by including categorical variables.  Suppose that we have a mixture of 
continuous and binary variables among the elements of $Z$ and $Y$.  Then
we might choose to model them using an approach analogous to that of 
\citet{fan17high}, who propose a Gaussian copula model that is dichotomized
for the binary components.  Their estimation methods show that the resulting
joint distribution is a smooth function of the parameters.  This model, 
combined with smooth marginal models will also be frugal and satisfy 
A\ref{ass:vi2}.  We use this approach in our data analysis
example in Section \ref{sec:data}.
\end{exm}


More general versions of the frugal parameterization are given in Sections 
\ref{sec:survival} and \ref{sec:snm}, though again we note that the rest of 
the paper can be read without reference to those sections. 

\section{Main Result} \label{sec:param}


We now give the main result outlined in the introduction:
given a weight function $w$,
a parameterization $\theta = (\theta_{\ZX}, \thetaYX, \phiYZIX)$ of $\pZXY$ 
induces a corresponding frugal parameterization $\theta^* = (\thetaZX, \theta_\YIX^*, \phi_\YZIX^*)$, also of $\pZXY$.  In particular, we can choose any parametric 
model for any cognate distribution $\pYX^*$, 
and use it to construct a smooth parameterization of the joint density. 
In other words, in terms of parameterization there is no essential difference between choosing a model for $\pYX^*$ or for the ordinary conditional distribution $\pYX$.
When we do this, the smoothness and regularity of the parameterization of the observational model ($\theta_\YIX$) as well as its variation independence
to $\thetaZX$ and---possibly---the association parameters, is preserved in the
new parameterization of $\pZXY$.
The Theorem \ref{thm:main} below formalizes this.

We first need to introduce a couple of additional assumptions.  Recall that the functionals $\thetaYX$ and $\phiYZIX$ for $\pZXY$ have an identical 
form to the functionals $\thetaYX^*$ and $\phiYZIX^*$ for $\pZXY^*$.
We will assume that $\pZX^* = \pX^* \cdot w$ is
smoothly and regularly parameterized by a function of $\theta_\ZX$,
and a relative positivity of the observational distribution.  Recall also that 
the analyst chooses $\pYX^*$ and $w$ based on subject matter considerations.\\[-18pt]
\begin{enumerate}[{A}1.]
\addtocounter{enumi}{1}
\item \label{ass:smoo2} 
The product $\pZX^* = \pX^* \cdot w$ has a smooth and regular parameterization
$\eta_\ZX := \eta_\ZX(\thetaZX)$, where $\eta_\ZX$ is a twice differentiable
function with a Jacobian of constant rank.

\item \label{ass:ac1} $\pZX$ is absolutely continuous with respect to $\pZX^*$ 
at the true distribution $\pZXY$.
\end{enumerate}

To clarify, we have two separate parameterizations of $\pZXY$.  The first, $\theta$,
corresponds to using the ordinary conditional distribution $\pYX$ in our frugal parameterization and `default' weight function $w_0(z \cmid x) = \pZIX(z \cmid x)$, whereas the second $\theta^*$ uses a cognate distribution $\pYX^*$ for some other 
weight $w$. 
As a note of caution, the two models for $\pZXY$ 
induced by $\theta$ and $\theta^*$ are \textbf{not} generally the same, because 
they apply to different functionals of $\pZXY$; if the models are both saturated 
then the sets of distributions themselves \emph{will} be the same, but the 
parameters have different interpretations, and their values are therefore generally
different. 

\vspace{6pt}

%
%

\begin{thm} \label{thm:main}
Let $\pZXY$ be a distribution parameterized by $\theta := (\thetaZX, \thetaYX, \phiYZIX)$
with weight function $\pZIX$,
and $w$ a kernel satisfying A\ref{ass:smoo2}; 
we also assume that A\ref{ass:ac1} holds.

Then $\theta$ is frugal w.r.t.~$\pYX$ 
if and only if $\theta^* := (\thetaZX, \thetaYX^*, \phiYZIX^*)$ is also frugal
w.r.t.~$\pYX^*$.
In addition, if $\phiYZIX$ satisfies A\ref{ass:vi2} and $\eta_\ZX(\Theta_\ZX)\subseteq \Theta_\ZX$, then $\phiYZIX^*$ also does.
\end{thm}

\begin{proof}
First, note that by definition, either parameterization can use 
$\thetaZX$ to obtain $\pZX$.  
Then combining with A\ref{ass:smoo2} we can obtain $w \cdot \pX^*$
as a smooth function of $\eta_\ZX(\thetaZX)$.  Then note that 
by A\ref{ass:ac1} we have
\begin{align}
\pZXY^*= \pZXY \frac{p_{\ZX}^*}{p_{\ZX}} = \pZXY \frac{w \cdot \pX^*}{p_{\ZX}}, \label{eqn:equiv}
\end{align}
so given that the fraction here is a smooth function of $\thetaZX$ from 
either parameterization, it is clear that we can obtain $\pZXY^*$ 
smoothly from $\theta^*$ if and only if we can obtain $\pZXY$ smoothly from
$\theta$.  
This proves that $\theta$ is a smooth and regular parameterization if and only if 
$\theta^*$ is.

For A\ref{ass:vi2}, note that if $\phiYZIX$ is variation independent of 
$\thetaZX$ and $\thetaYX$, then 
we also have that $\phiYZIX^*$ is variation independent of $\eta_\ZX(\thetaZX)$ and $\thetaYX^*$,
because this is just A\ref{ass:vi2}
applied to the (possibly) smaller set of distributions $\pZXY^*$.  
Then notice that modifying the value of $\thetaZX$ in such a way that keeps 
the value of $\eta_\ZX$ the same will have no effect on the possible values 
of $\phiYZIX^*$, and hence A\ref{ass:vi2} holds for $\theta^*$.
%
\end{proof}


\begin{rmk}
The previous result tells us that, given a suitable dependence measure $\phi$,
we can propose almost arbitrary (i.e.\ provided that they satisfy the assumptions
indicated in the Theorem) separate parametric models for
each of the three quantities $\pZX(z, x)$, $\pYX^*(y \cmid x)$ and
$\phiYZIX^*(y, z \cmid x)$,
and be sure that there exists a (unique) joint distribution $\pZXY(z,x,y)$
compatible with that collection of models.
Of course, this leaves open the question
of how we should compute that joint distribution.

The requirement that the image of $\eta_\ZX$ is contained within the set of possible
distributions $\pZX$ is a very mild condition.
In addition, if we use a copula or odds ratio as the conditional association measure 
the implication \emph{always} holds, regardless of this assumption.
%
\end{rmk}

\begin{exmp} \label{exm:run2}
Picking up Example R\ref{exm:run} again and, for now, consider only the observed
variables (though see Example R\ref{exm:havercroft_sim} in Appendix
\ref{sec:vcop} for details on how to
simulate from all the variables).  Take $Z = L$ and $X = (A,B)$, then
Theorem \ref{thm:main} says that we can parameterize the model using
parametric models of the three pieces
\begin{align}
&p_{\ALB}(a, \ell, b) && p_{\YIAB}(y \cmid \Do(a,b)) && \phi_{L\YIAB}^*(\ell, y \cmid a,b). \label{eqn:havercroft_param}
\end{align}
For convenience, we choose to factorize $p_{\ALB}$ according to the ordering $A,L,B$.
Set $A \sim \operatorname{Bernoulli}(\theta_a)$, $L$ is conditionally exponentially
distributed with mean $\E[L \mid A=a] = \exp(-(\alpha_0 + \alpha_a a))$, and
\begin{align*}
B \mid A=a, L=\ell \sim
\operatorname{Bernoulli}(\expit(\gamma_0 + \gamma_a a + \gamma_\ell \ell + \gamma_{a\ell} a \ell)).
\end{align*}
Let us suppose that $Y$ is normally distributed under the intervention on $A,B$,
with mean
\begin{align*}
\E[Y \mid \Do(A=a, B=b)] &= \beta_0 + \beta_a a + \beta_b b + \beta_{ab} ab
\end{align*}
and variance $\sigma^2$.
Let
$\phi_{L\YIAB}^*$ be a conditionally bivariate Gaussian copula, with correlation
parameter given by some function $\rho_{ab}$ of $a$ and $b$.  This parameterization
is frugal and satisfies A\ref{ass:vi2}.

In addition, note that this approach entirely circumvents the
g-null paradox discussed in Example R\ref{exm:gnull}, because the marginal
dependence of $Y$ on $A$ (after intervention on $A$ and $B$) is uniquely and
explicitly encoded by the parameters $\beta_a, \beta_{ab}$.  
\end{exmp}

\section{Sampling from a marginal causal model} \label{sec:sim}

In this section we will consider how to sample from $\pZXY$ using
a frugal parameterization $\theta^*$, sometimes analytically,
but more commonly via the method of rejection sampling.
Note that, now we have constructed a valid parameterization, we
will no longer need to refer to the model on $\pZXY$ defined by $\theta$.
From this point on, we only discuss the model on $\pZXY$ parameterized by $\theta^*$,
and the corresponding model on $\pZXY^*$ that replaces $\thetaZX$ with 
$\eta_\ZX(\thetaZX)$.

We first
review how one should go about choosing such a parameterization.

\begin{enumerate}
\item Choose the quantity $\pYX^*$ which you wish to model, or of which you
wish to model a function, and select a parameterization $\thetaYX^*$ (this should
include the quantity of interest).

\item Determine the kernel $w$ over which we need to integrate $\pYIZX$ in 
order to obtain $\pYX^*$, and a dummy marginal distribution $\pX^*$ over $X$.  
This should not be degenerate, and for efficient sampling should be similar in form to 
the observational margin $\pX$.

\item Introduce a parameterization $\thetaZX$ of $\pZX$, such that
$\pZX^* = w \cdot \pX^*$ is smoothly and regularly parameterized by a
twice differentiable function $\eta_\ZX$ of $\thetaZX$.

\item Choose a `suitable' parameterization $\phiYZIX^*$ of the dependence
in $Z$-$Y$ conditional upon $X$ in the causal distribution $p^*$.
\end{enumerate}

The three pieces $\thetaZX$, $\thetaYX^*$ and $\phiYZIX^*$ will make up the
frugal parameterization.  To make point 3 more concrete, in
Example \ref{exm:binYZ} we can
set $\thetaZX$ to be the combination $(q, \gamma, \sigma^2)$, and then
take $\pX^* \sim N(0, 2\sigma^2)$; this ensures it will have heavier tails
than $p_\XIZ \sim N(\gamma z, \sigma^2)$ which, as we will see in Section
\ref{sec:alg}, is crucial for sampling.

For point 4, the question of suitability of the dependence measure, we would wish to
consider: (i) whether the relevant variables can be modelled with the
particular dependence measure selected (e.g.\ odds ratios are suitable for discrete
variables, but not so useful in practice for continuous ones);
(ii) the computational cost of constructing the joint distribution;  (iii)
whether we want the dependence measure to be variation independent of
its baseline measure; if so that would rule out risk ratios and differences.
For a larger model with a vector valued $X$, we might wish to fit different dependence
measures for each treatment variable; see Section \ref{sec:snm} for an example
of this with a Structural Nested Mean Model.

\subsection{Direct Sampling} \label{sec:or}

For fully discrete or multivariate Gaussian models, it is possible
to compute $\pZXY^*$ and then  `reweight' by $\pZX/\pZX^*$ to
obtain the distribution $\pZXY$ in closed form.  As noted in Proposition \ref{prop:disc},
in the discrete case this is straightforward using (conditional) log
odds ratios to obtain a frugal parameterization of the distributions.
For example, if $Y$ and $Z$ are both binary, taking values in $\{0,1\}$, we can use
\begin{align*}
\log \phiYZIX(x) := 
\log \frac{\pYZX(1,1 \cmid x) \cdot \pYZX(0,0 \cmid x)}{\pYZX(1,0 \cmid x) \cdot \pYZX(0,1 \cmid x)}.
\end{align*}
For further details, including what happens if there are more than two levels to $Y$ or $Z$, see \citet{br02}.
%
As noted in (\ref{eqn:or_inv}), a nice property of the odds
ratios as the association parameter is that their
values in the observational and causal distributions are always the same.

\begin{exmp} \label{exm:run3}
Let us apply this to a discrete version of Example R\ref{exm:run2} from
\citet{havercroft12}; we know that the objects in (\ref{eqn:havercroft_param})
are sufficient to define the model of interest.  If all the variables are
binary, then we start with a parameterization of $p_{\ALB}$ and $p^*_{\YIAB}$
using (conditional) probabilities, and $\phi^*_{L\YIAB} (=\phi_{L\YIAB})$ using conditional
odds ratios.

Assume, for example, that
\begin{align*}
Y \mid \Do(A=a, B=b) \sim \operatorname{Bernoulli}(\expit(-1 + a + ab)),
\end{align*}
with $L \mid A=a \sim \operatorname{Bernoulli}(\expit(2a - 1))$,
and $\log \phi_{L\YIAB}(a,b) = 1 + a - 2b + ab$.  Then specifying, for instance,
$B \mid A=a, L=\ell \sim \operatorname{Bernoulli}(\expit(1-a-2\ell+a\ell))$
implies that the  ordinary conditional $p_{\YIAB}$
is, by a direct calculation,
\begin{align*}
Y \mid A=a, B=b \sim \operatorname{Bernoulli}(\expit(-0.245 + 0.432a - 0.500b + 0.846ab)).
\end{align*}
Note that the `observational' conditional parameters are quite different from
their causal counterparts.
%
\end{exmp}

\subsection{Sampling By Rejection} \label{sec:alg}

In most realistic situations the data cannot be modelled as entirely discrete
or multivariate Gaussian.  In such cases we suggest simulating from 
a distribution constructed analogously to 
the causal model, and then using
rejection sampling to modify the marginal distribution of $X$
and $Z$ and obtain data from the corresponding observational distribution. 
The idea of rejection sampling is very simple.  Suppose we have two distributions: a \emph{target} $p$ that is difficult to sample from, and a \emph{proposal}
$q$ that is both easy to sample and \emph{dominates} $p$, in the sense that
there is some $M$ such that $p/q \leq M$ in a $p$-almost sure sense; then we 
can obtain independent samples from $q$ and reject only those samples $X$ for 
which $p(X)/q(X) > M \cdot U$, where $U$ is an independent uniform random variable 
on $(0,1)$. 
The samples that are not rejected are then distributed independently 
from $p$ \citep[see, for example,][Chapter 2]{robert:04}. 

We might hope that, since $\pZXY^*$ is relatively easy to sample from, then 
we would find that $\pZX^* = w \cdot \pX^*$ dominates $\pZX$; unfortunately this 
is generally not the case and is extremely implausible unless 
$Z$ is discrete.  However, a weaker assumption is sufficient.  

\begin{enumerate}[{A}1.]
\addtocounter{enumi}{3}
\item \label{ass:tails} The set $\Z$ can be partitioned into a countable number of
bins $\mathcal{B}=\{B_i\}$ such that, for each $i$, there $p_\ZX$-almost surely exists $M_i$ with 
$\pZX(z, x)/\pZX^*(z, x) \leq M_i$ for all $x\in \X, z \in B_i$. 
\end{enumerate}
The significance of this assumption is that given $n$ i.i.d.~realizations from $\pZ$ 
we can then partition them into $\mathcal{B}$, and target obtaining the same number 
of observations via a local rejection sampling scheme in each bin.  Note that this
original sample of $Z$s is never used after determining the number of observations
within each bin. 

Of course, to use this assumption we must be able to sample from $\pZXY^*$, 
and the feasibility of this depends 
upon the particular model; however it is generally a much easier condition to 
satisfy than being able to sample from $\pZXY$ directly given that $\pYX^*$ is
already specified.  With a copula, 
it is essentially trivial: we can just sample directly from the 
copula, and then use inversion to ensure the margins are correct \citep{clifford94monte}.

We note that if the weighting is sometimes particularly heavy or the model
is high-dimensional, then some of bounding constants $M_i$ will be large and/or
some of the bins for $Z$ have very low probability of being proposed, 
so the rejection method becomes very inefficient. 
However, since we can evaluate the joint distribution exactly if we use a copula, 
other more advanced simulation
methods can be used instead of rejection sampling.  A disadvantage is that the  
samples would generally only be approximately distributed correctly, but the level 
of error could easily be chosen to be statistically undetectable.  We leave this to 
future work.

\subsection{Copulas} \label{sec:cop}

As previously discussed, copulas may provide an approach
to a frugal parameterization of models with continuous $Y$ and $Z$.
In this section we describe how copulas may be used to simulate
from and fit causal models with particular marginal specifications.

In the simplest case, we can start by simulating values for $X$ using some $\pX^*$,
and then use the copula to simulate from 
the causal distribution on the scale of quantiles.  We then 
apply the inverse CDF of $\pYX^*(y \cmid X=x_i)$ and $\pZ(z)$ to the uniform 
margins to obtain the actual observations.  The parameters for the copula 
itself (i.e.~$\phiYZIX^*$) may or may not depend upon $X$. 
To obtain samples from the observational distribution, we can use rejection
sampling, provided that A\ref{ass:tails} is satisfied.


\begin{exmp} \label{exm:run4}
Continuing our running example from \citet{havercroft12}, suppose we now wish
to simulate some data from the model specified in Example R\ref{exm:run2} by rejection sampling.
We first select some  values for the parameters:
\begin{align*}
\theta_a &= 0.5 & (\gamma_0, \gamma_a, \gamma_\ell, \gamma_{a\ell}) &= (-0.3, 0.4, 0.3,0)  \\
(\alpha_0, \alpha_1) &= (0.3,-0.2) & (\beta_0, \beta_a, \beta_b, \beta_{ab}) &= (-0.5, 0.2, 0.3,0)
\end{align*}
and $\rho_{ab} = 2\expit(1+a/2)-1$.
Taking a
large sample size of $10^6$, we indeed find (empirically, using goodness-of-fit tests)
that $\E A = 0.5$,
that $L$ appears to be exponentially distributed with the specified mean, and that
$\E [B \mid A=a, L=\ell]$ has the correct form.  In addition, if we fit
an inverse probability weighted (IPW) linear model for $Y$ (using the fitted
value we obtain from the regression for $B$, see \citealt[][Chapter 12]{hernan:20}) the parameters for the interventional distribution of $Y$ under $\Do(A=a, B=b)$ are also as expected:
\begin{align*}
\hat\beta_0 &= -0.4985 \,(0.0022)&
\hat\beta_a &= 0.1999  \,(0.0033) &
\hat\beta_b &= 0.3002  \,(0.0030)&
\hat\beta_{ab} &=  -0.0030 \,(0.0042). 
\end{align*}
Code to replicate this analysis can be found in the vignette \texttt{Comparison}
of the R package \texttt{causl} \citep{causl:21}.
\end{exmp}

Copulas lack many of the attractive properties of odds ratios, such as 
the invariance in (\ref{eqn:or_inv}), and their interpretation is different 
because it is in terms of the quantiles of the margins
rather than their actual value.  However, they can
be extremely flexible if one has a multivariate outcome, because one
can make use of \emph{vine copulas} to model them.  See Appendix 
\ref{sec:vcop} for more details.

\section{Fitting Methods} \label{sec:fitting}

We start this subsection with a result telling us how to fit marginal
structural models using maximum likelihood (ML) estimation.  In fact, it
turns out that if we have a marginal structural model and our full model
parameterized by $\theta^*$ is correctly specified for the \emph{observational}
data from $\pZXY$,
then the MLE for $\pYX^*$ is obtained by maximizing the likelihood for the
\emph{causal} model (i.e.\ with $X$ and $Z$
assumed to be independent) with respect to the observational data from $\pZXY$
(so $X$ and $Z$ are in fact \emph{not} independent).
This is the content of Theorem \ref{thm:fitting} below.

Note also that, although this result will not generally hold if
part of the model is misspecified, if the \emph{propensity score} model
$p_{\XIZ}$ is incorrect then this will not affect inference about the
remainder of the model when $w(z) = \pZ(z)$.
This is because there is a parameter cut
between $\pZ \cdot p_{\YIZX}$ and $p_{\XIZ}$ \citep[see, e.g.][]{bn:78},
and the parameters $\thetaYX^*$ and $\phiYZIX^*$ are (for MSMs)
functions of this first quantity.  



For results connected with fitting, we will assume that all
our parameters are identifiable from the available data (cf.\ Remark \ref{rmk:identifassump}).  In particular,
we will also make use of A\ref{ass:ac1} again, since we cannot hope to
recover a distribution that does not satisfy a positivity assumption.  Since
the result concerns maximum likelihood estimation, we will make the very
slightly stronger assumption that the Kullback-Leibler divergence between
$p$ and $p^*$ is finite.  (Note that this is a strictly weaker assumption
than A\ref{ass:tails}.)
\begin{enumerate}[{A}1.]
\addtocounter{enumi}{4}

\item \label{ass:kl} $\KL(\pZX \operatorname{\|} \pZX^*) := \E_{\pZX} \log \frac{\pZX(Z,X)}{\pZX^*(Z,X)} < \infty$.
\end{enumerate}
We refer to the parameters of the causal parameterization of
the observational distribution as $\theta^* = (\thetaZX, \thetaYX^*, \phiYZIX^*)$,
and of the causal distribution as $\eta(\theta^*) := (\eta_\ZX(\thetaZX), \thetaYX^*, \phiYZIX^*)$.

\begin{thm} \label{thm:fitting}
Suppose that $\theta^*$ is a frugal parameterization with weight function 
$w(z) = \pZ(z)$,
so the model we are
interested in is the marginal structural model; suppose also that A\ref{ass:kl} holds.
The maximum likelihood estimator $\hat \eta$ of $\eta(\theta^*)$ obtained with
the observed data (i.e.\ data generated using the distribution $\pZXY$ with
parameters $\theta^* = (\theta_{\ZX},\thetaYX^*, \phiYZIX^*)$)
will be consistent for the
distribution in the causal model with parameters $\eta = (\eta_\ZX, \thetaYX^*, \phi_\YZIX^*)$.


In addition, for the estimates obtained in this way, we have
\begin{align*}
\sqrt{n}\left\{\left(\begin{matrix}
\hat{\theta}_\YIX^* \\
\hat{\phi}_\YZIX^*\end{matrix}\right) -
\left(\begin{matrix}
{\theta}_\YIX^* \\
{\phi}_\YZIX^*\end{matrix}\right)\right\}
\stackrel{d}{\longrightarrow} \;  N \! \left(0, \, I(\theta^*)_{\thetaYX^*, \phiYZIX^*}^{-1} \right),
\end{align*}
where $I(\theta^*)$ is the Fisher information under $\pZXY$ and
$I(\theta^*)_{\thetaYX^*, \phiYZIX^*}^{-1}$ is the submatrix of its inverse
relating to $\thetaYX^*$ and $\phiYZIX^*$.
\end{thm}

\begin{proof}
\citet[][Lemma 5.35]{vandervaart:98} shows that if the target distribution is identifiable,
then maximum likelihood estimation converges to the KL-minimizing distribution.
Consider the density for the causal model:
\begin{align*}
\pZXY^*(z, x, y) &= \pX^*(x) \, w(z) \, \pYIZX(y \cmid z, x),
\end{align*}
where we suppress dependence upon parameters.
For a comparison with the density of the data, note that
\begin{align*}
\frac{\pZXY(z, x, y)}{\pZXY^*(z, x, y)} 
= \frac{\pZX(z, x)}{\pX^*(x) \cdot w(z)} = \frac{\pZX(z, x)}{\pZX^*(z,x)},
\end{align*}
and hence the KL-divergence is finite by A\ref{ass:kl}.  Then,
\begin{align*}
\KL(\pZXY \operatorname{\|} \pZXY^*) &= \int_{\Z\X\Y} \pZXY(z,x,y) \log \frac{\pZX(z, x)}{\pZX^*(z, x)} \, dz \,dx \, dy\\
&= \int_{\Z\X} \pZX(z,x) \log \frac{\pZX(z, x)}{\pZX^*(z, x)} \, dz \,dx\\
&= \KL(\pZX \operatorname{\|} \pZX^*).
\end{align*}
Now, in general the result of minimizing this expression will depend upon
the precise parameterization of $\pZX^*$, but the minimization will
pick out the distribution that is `closest' to $\pZX$ within the causal
model.  The result for marginal structural models is a consequence
of the fact that the minimizing distribution in this case is $\pX \cdot \pZ$.
%

For the asymptotic distribution of the estimators $\hat\theta_\YIX^*$ and
$\hat\phi_\YZIX^*$, notice that for a marginal structural model we have $\pZ = \pZ^*$
(and of course we always have $p_\YIXZ=p_\YIXZ^*$) so the parameter cut mentioned
above applies to both models.  Hence, there is no asymptotic correlation between
$(\hat\theta_\YIX^*, \hat\phi_\YZIX^*)$ and $\hat\theta_\ZX$ (or $\hat\eta_\ZX$).
Then the asymptotic variance is just a
standard result for MLEs \citep[see, e.g.][Chapter 18]{ferguson:96}.  
\end{proof}

Note
that we \emph{must} apply the Fisher information under $\pZXY$ in order to
obtain the correct variance, since this is the distribution of the data being
used to approximate the expectation.
While the proof above is stated only for the single time-point exposure model, 
it extends to a longitudinal case with multiple treatments, 
similar to the obvious extension of the model in our running example.

When computing standard errors in practice we use the observed information (i.e.\ an
empirical approximation to the Fisher Information), rather than its theoretical
mean. 
In principle we could also use a `sandwich estimate' to obtain more robust
standard errors; because we know that our models are correct we do not do this,
but for other users of this method on real data we would always recommend
using sandwich errors.  In our case these would be the square-roots of the 
diagonal entries of
\begin{align*}
B(\theta^*)^{-1} A(\theta^*) B(\theta^*)^{-1},
\end{align*}
where
\begin{align*}
&& A(\theta^*) &= \E_{\theta^*} \frac{\partial \ell}{\partial \eta}\frac{\partial \ell}{\partial \eta}^T &\text{and}&& B(\theta^*) &= \E_{\theta^*} \frac{\partial^2 \ell}{\partial \eta^2}. &&
\end{align*}
Note that although this result shows that we \emph{can} fit models via
maximum likelihood estimation, if the model is misspecified
there is no guarantee that the estimator will be consistent or even close 
to the true value. 
Other less sensitive estimators, such as doubly robust approaches (see 
Remark \ref{rmk:dr} below), may therefore be more useful in practice than the MLE.

\begin{rmk}
Note that the same result (i.e.\ convergence of the estimator to
the KL closest distribution to $\pZ$) will hold for the ETT estimator with kernel
$w(z) = p_\ZIX(z \cmid 1)$, since this is also independent of the value of $X$.  In order
to estimate the parameters for this kernel, we would have to consider the
subset of data for which $X$ takes the particular value 1.  We could then
obtain an MLE for the whole model by combining the complete data estimator 
with the separate estimate for $w$ obtained from the treated patients.
\end{rmk}

\begin{rmk}
Given maximum likelihood estimates for the parameterization of $\pZXY$, we 
can of course use the invariance properties of MLEs together with 
the delta method for the standard errors, to obtain an estimate for any 
(differentiable) function of the parameters that we choose. 
\end{rmk}

\begin{rmk} \label{rmk:dr}
Taking a \emph{doubly robust} approach to estimating the causal 
parameters,  we see that if $\phiYZIX(y,z\cmid x) := c_{UV|X}(F_{\YIX}(y \cmid x), F_{\ZIX}(z \cmid x) \cmid x)$ is a copula density, then
\begin{align*}
\pZ(z) \cdot \pYIZX(y \cmid z, x) &= \pZ(z) \cdot \pYX^*(y \cmid x) \cdot \phiYZIX^*(y, z \cmid x)\\
\hspace{-4.3cm} \text{and therefore} \hspace{2.2cm}
\pYIZX(y \cmid z, x) &= \pYX^*(y \cmid x) \cdot \phiYZIX^*(y, z \cmid x),
\end{align*}
so $\hat{Q}(z,x) = \E[Y \cmid Z=z, X=x]$ can fairly easily be computed numerically; 
indeed, if $Y$ and the copula are both Gaussian, we obtain it in closed form.  
We then fit a model, say $\hat\pi(x \cmid z)$, for the propensity score $p_\XIZ(x \cmid z)$.  

A doubly robust estimator uses $\hat{Q}$ and $\hat{\pi}$ to construct 
an estimating equation, and will give a 
consistent estimate for the causal parameter if either model
is correctly specified.
If they are both correct, then 
this estimator is also semiparametric efficient \citep{scharfstein:99}.
Using a doubly robust approach to compare with the MLE will help to protect 
us against possible misspecification of $\phiYZIX^*$; this is useful given 
that choosing the association parameter is not particularly intuitive.
\end{rmk}

\subsection{Simulation} \label{sec:sim_study}

We now run a simulation to compare four methods:
outcome regression, inverse probability weighting, our maximum likelihood
estimation, and standard doubly robust estimation (i.e.~just using an
ordinary regression model, not as described in Remark \ref{rmk:dr}). 

We use the setup described in Examples R\ref{exm:run2} and R\ref{exm:run4} (Sections \ref{sec:param} and \ref{sec:cop} respectively) to generate our
data, so again $Y$ (after intervening to set $\{A=a,B=b\}$) is normally distributed with
mean $-0.5 + 0.2a + 0.3b$ and variance 1.
We then performed $N=1\,000$ runs of the analysis above with sample size $n=250$.
The results are shown in Table \ref{tab:comp2}, with boxplots of the biases
in Figure \ref{fig:bias_plots}.  The table contains the average bias, the
empirical coverage of a 90\% interval, and the \emph{standard error calibration},
which we define as:
\begin{align*}
\operatorname{sec} = \left(\frac{1}{N} \sum_{i=1}^N \frac{\operatorname{bias}(\hat\theta_i-\theta)^2}{\operatorname{se}(\hat\theta_i)^2} \right)^{1/2}.
\end{align*}
If this value is less than one it suggests that the standard errors are
conservative, if larger than one it suggest they are too small.

Outcome regression performs poorly, although this is to be expected as the model is misspecified. 
We see that the other three
methods all have very comparable performance and efficiencies, and are
mostly well calibrated: the MLE and DR methods give slight under coverage
for the first two parameters, though the double robust method gives
conservative standard errors for the interaction parameter.
An example on a larger simulated dataset is given in Appendix \ref{sec:sim2}.

\begin{table}
\begin{center}
\begin{tabular}{lrrrrrr}
\toprule
\multicolumn{1}{c}{ } & \multicolumn{3}{c}{Outcome Reg.} & \multicolumn{3}{c}{IP Weighting} \\
\cmidrule(l{3pt}r{3pt}){2-4} \cmidrule(l{3pt}r{3pt}){5-7}
coef & bias & cover90 & se calib & bias & cover90 & se calib\\
\midrule
1 & $-$0.0769 & 0.837 & 1.21 & 0.0038 & 0.905 & 0.99\\
a & $-$0.0303 & 0.880 & 1.03 & $-$0.0096 & 0.932 & 0.93\\
b & 0.1538 & 0.755 & 1.33 & $-$0.0018 & 0.935 & 0.91\\
a.b & 0.0220 & 0.901 & 1.00 & 0.0038 & 0.942 & 0.85\\
\midrule
\multicolumn{1}{c}{ } & \multicolumn{3}{c}{Double Robust} & \multicolumn{3}{c}{MLE} \\
\cmidrule(l{3pt}r{3pt}){2-4} \cmidrule(l{3pt}r{3pt}){5-7}
coef & bias & cover90 & se calib & bias & cover90 & se calib\\
\midrule
1 & 0.0046 & 0.879 & 1.06 & 0.0046 & 0.882 & 1.06\\
a & $-$0.0098 & 0.876 & 1.08 & $-$0.0071 & 0.891 & 1.03\\
b & $-$0.0014 & 0.919 & 0.97 & $-$0.0026 & 0.898 & 1.02\\
a.b & 0.0054 & 0.982 & 0.69 & 0.0040 & 0.893 & 1.01\\
\bottomrule
\end{tabular}
\caption{Table giving the average bias, coverage of a 90\% confidence interval,
and standard error calibration (the ratio of absolute bias to standard error)
of four methods:
outcome regression; inverse probability (IP) weighting; a doubly robust estimator; 
and maximum likelihood estimation (MLE).}
\label{tab:comp2}
\end{center}
\end{table}

\begin{figure}
\includegraphics[width=15cm]{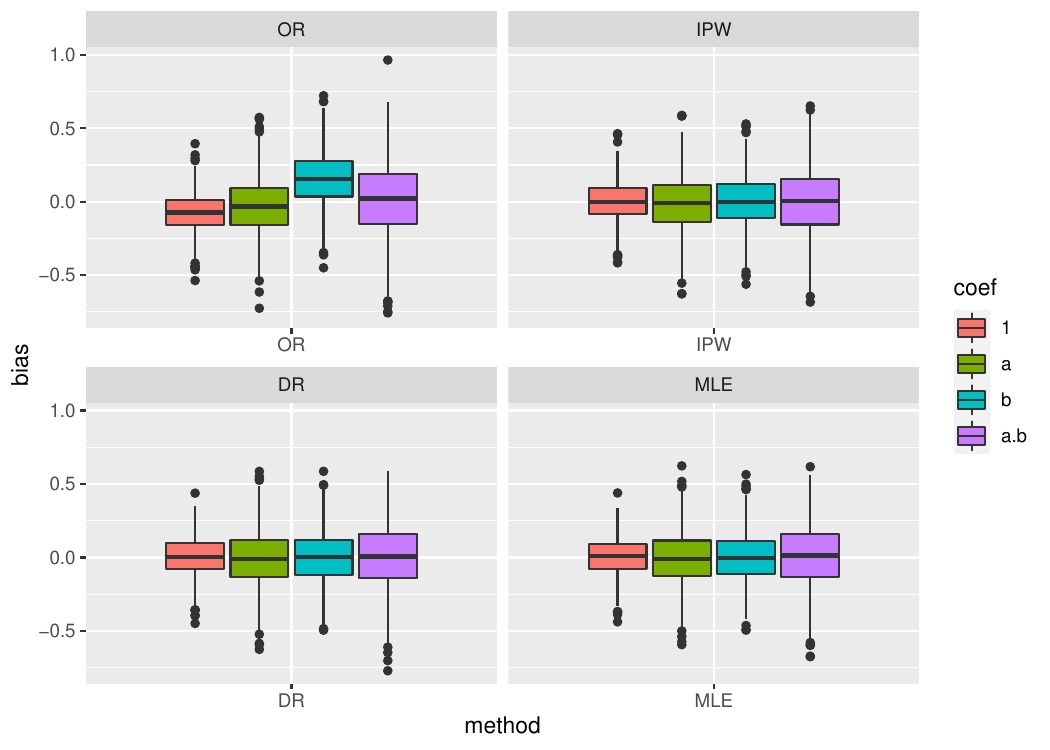}

\caption{Boxplots of the bias for each coefficient by four methods: 
outcome regression (OR), inverse probability weighting (IPW), 
doubly robust estimation (DR), and maximum likelihood estimation (MLE).}
\label{fig:bias_plots}
\end{figure}

\subsection{Data Analysis} \label{sec:data}


To illustrate our method, we apply the maximum likelihood fitting
procedure to data from the IDEFICS study \citep{ahrens11idefics}.  The subset of data we use consists of
measurements of 531 German children aged between
2 and 9, including their sex, physical activity, screen time, parental
education, a `vegetable score', fibre intake, and a polygenic
risk score (PRS) for BMI.
The study also records the child's BMI and their parents' BMIs.
Preliminary analyses suggest that increased fibre intake can reduce
BMI, especially for those children who have a strong genetic predisposition
for obesity \citep{huls2021}.
Our aim is to study the effect modification of PRS on the
relationship between fibre intake and actual BMI, whilst
adjusting for confounding due to other covariates.


We replicate the setting in \citet{nohren:msc}, which considers how the 
causal effect of a dichotomized indicator of fibre intake ($X$)
on age and sex standardized BMI ($Y$, a z-score) interacts with
the dichotomized polygenic risk score ($C$); like \citeauthor{nohren:msc}
we also use a marginal structural model:
\begin{align*}
\E [Y \cmid C=c; \Do(X=x)] = \beta_0 + \beta_1 x + \beta_2 c + \beta_3 cx.
\end{align*}
We assume that all other variables are causally prior to
$X$, so that $\E [Y \cmid C=c; \Do(X=x)]$ is our causal
distribution of interest,
where $\bs Z$ consists of other confounders; these include sex, age,
physical activity, screen time, vegetable score and a dichotomized
version of parental education level.  

We choose an ordinary Gaussian linear model for the MSM,
and also the other models used for variables in $\bs Z$.  The
copula was also Gaussian.  Note that in order to accommodate sex
and parental education as binary
variables it was necessary to integrate over the copula, 
effectively making it a probit model (see also Example \ref{exm:mix}).

The relevant coefficients from the model fit are shown in Table
\ref{tab:BIPS}.
Under our modelling assumptions, these results do not suggest that 
increased fibre intake reduces BMI and thus we cannot confirm previous
results obtained on a larger dataset of 2,688 children from seven countries
\citep{nohren:msc}; the estimates from that study are within
our (rather wide) confidence intervals, though.  Furthermore, the 
analysis of \citeauthor{nohren:msc} for the marginal structural
model using inverse probability weighting on only the German data yields 
slightly different parameter
estimates, and larger standard errors (see Appendix \ref{sec:data2}
for details). This illustrates---in a practical analysis---the differences 
between, on the one hand, modelling the $C$-$Y$-$\bs Z$ association directly or, 
on the other hand, modelling the propensity score for the inverse 
probability weights.

\begin{table}
\begin{center}
\begin{tabular}{cc|cc|cc}
\toprule
param. & coefficient & est. & s.e. & \multicolumn{2}{c}{95\% conf.\ int.} \\
\midrule
$\beta_1$ & fibre & 0.049 & 0.092 & $-$0.132 & 0.230 \\
$\beta_2$ & PRS & 0.374 & 0.198 & $-$0.015 & 0.762 \\
$\beta_3$ & PRS:fibre & 0.011 & 0.359 & $-$0.693 & 0.715 \\
\bottomrule
\end{tabular}
\caption{Table giving estimated coefficients in the marginal structural
model for effect modification of the PRS on BMI by fibre intake.}
\label{tab:BIPS}
\end{center}
\end{table}

\section{Survival Models} \label{sec:survival}

Another application of the frugal parameterization is to causal longitudinal models, 
and in particular to survival models. 
Note that with sequences of treatment variables, the sequential versions of identifying assumptions must be met (cf.\ Remark \ref{rmk:identifassump}) known as \emph{sequential conditional exchangeability}; we continue to take these as given in Sections 6 and 7.

The following corollary of Theorem \ref{thm:main} allows us to `build up' a frugal
parameterization of the joint distribution using several different cognate quantities. 
Given a collection of variables $Y_1, \ldots, Y_d$ under 
some natural ordering (typically a temporal ordering),
let $[i-1] = \{1, \ldots, i-1\}$ denote the
predecessors of each $i =1, \ldots, d$.

\begin{cor} \label{cor:build}
Let $Y_1, \ldots, Y_d$ have joint density $p$, and let 
$X_i:=Y_{A_i} = \{Y_j : j \in A_i\}$ for
some $A_i \subseteq [i-1]$.
Also let
$p_{Y_i|X_i}^*(y_i \cmid x_{i})$ be defined by applying (\ref{eqn:cognate}) 
with $Z_i := Y_{L_i}$ where $L_i := [i-1]\setminus A_i$. 

Then there is a smooth and regular parameterization of the joint distribution,
which can be chosen to be variation independent, containing each 
$p_{Y_i|X_i}^*(y_i \cmid x_{i})$.
\end{cor}

\begin{proof}
We proceed by induction.  For $i=1$ we just have a smooth and regular parameterization
of $p_{Y_1}(y_1)$.  For a general $i$, assume we have a smooth parameterization of
the joint density for $Y_1,\ldots, Y_{i-1}$ and of $p_{Y_i|X_i}^*(y_i \cmid x_{i})$.  Then 
using some appropriate $\phi_{Y_iZ_i \mid X_i}^*$ to make up a frugal 
parameterization (and A\ref{ass:vi2} if required),
by Theorem \ref{thm:main} we obtain a smooth and regular parameterization 
of the joint density for $Y_1,\ldots, Y_{i}$, and---if A\ref{ass:vi2} holds---the 
quantities used are all variation independent of one another.
\end{proof}

We refer to this approach as a \emph{recursive} or \emph{nested}
frugal parameterization, because in each case `the past' (i.e.~$\pZX$)
is itself parameterized in a frugal manner. 






\begin{exm} \label{exm:young}
\citet{young:14} consider survival models with
time-varying covariates and treatments.  Let
$Y_t = 0$ be an indicator of survival up to time
$t$ (with $Y_t = 1$ indicating failure).  Let
$L_t, A_t$ be respectively covariates and treatment
at time $t = 0, \ldots, T$.
 \citeauthor{young:14} model the quantities
\begin{align*}
P(Y_t = 0 \cmid Y_{t-1} = 0; \Do(a_1, \ldots, a_{t-1})), && t=1,\ldots,T;
\end{align*}
i.e.\ probability of survival to the next time point given treatment history
and survival so far. 
Under their assumptions these quantities are
identifiable via the g-formula as
\begin{align}
p(y_t \cmid y_{t-1}; \Do(\overline{a}_{t-1})) = \sum_{\overline{\ell}_t} p(y_t \cmid y_{t-1}, \overline{a}_{t-1}, \overline{\ell}_t) \prod_{s=1}^t p(\ell_s \cmid \overline{a}_{s-1}, \overline{\ell}_{s-1}); \label{eqn:gform_yt}
\end{align}
note that we omit some subscripts on densities for brevity.
Corollary \ref{cor:build} tells us that, setting $X_t = \overline{A}_t$ and
$Z_t = \overline{L}_t$, a parameterization exists of
the joint distribution that uses these quantities for each $t=1,\ldots,T$.
Given the distribution of $p(\overline{a}_{t-1}, \overline{\ell}_{t-1}, \overline{y}_{t-1})$, the
quantities
\begin{align*}
p(y_t \cmid y_{t-1}; \Do(\overline{a}_{t-1})) && \phi^*_{Y_t \overline{L}_{t-1} | \overline{A}_{t-1} \overline{Y}_{t-1}}(y_t, \overline{\ell}_{t-1} \cmid \overline{a}_{t-1}, \overline{y}_{t-1})
\end{align*}
may be used to recover $p(\overline{a}_{t-1}, \overline{\ell}_{t-1}, \overline{y}_{t})$.
\end{exm}

\citet{young:14} note that simulation from this model is difficult for certain
parametric choices, because some parameters from the joint model
and the marginal model are tied together in complicated ways.
They derive results that allow them to compute
particular causal parameters as functions of the joint distribution, and
hence to evaluate the performance of simulation methods exactly.
Our approach overcomes this problem by allowing causal quantities
of interest to be specified explicitly, and then have the rest of the
distribution constructed around them.

The model is parameterized so that failure is a rare outcome,
which allows approximation of the expit function by an exponential
function.  The parameters of interest are then those of the Cox
Marginal Structural Model:
\begin{align*}
\frac{p_{Y_t|\overline{A}_tY_{t-1}}^*(1 \cmid \Do(\overline{a}_{t}), Y_{t-1} = 0)
}{p_{Y_t|\overline{A}_tY_{t-1}}^*(1 \cmid \Do(\overline{0}_{t}), Y_{t-1} = 0)} &= e^{\gamma(t, \overline{a}_t)} = \exp\left(\psi_0 a_t + \psi_1 a_{t-1} + \psi_{01} a_t a_{t-1}  \right), 
\end{align*}
which, as we  see above, the authors assume to depend only upon the previous
two treatments.
These parameters $\psi$ are estimated by fitting an inverse weighted GLM to the data.

%

The authors also state that:
`[we] therefore, may be limited to simulation scenarios with the
proposed algorithm to particularly unrealistic settings if we wish
simultaneously to generate data under the null.'  Our results
demonstrate that if one uses our algorithms this is \emph{not} the case.
The null in this example corresponds to $\psi_0 = \psi_1 = \psi_{01} = 0$;
since the model is discrete we are free to choose arbitrary regression models
for the treatment on the observed past, for the covariates on their past
values and treatments (and even unobserved quantities), and
any arbitrary dependence structure between survival and the covariates,
conditional on all previous treatments and covariates.  This will allow
us to simulate from \emph{any} distribution under which treatment has no
(marginal) causal effect upon survival.
In Appendix \ref{sec:young} we perform some simulations on this
model.

\section{Structural Nested Model Parameterizations} \label{sec:snm}

Not all causal parameterizations involve modelling the entire conditional 
distribution for every level of the conditioning variable; i.e.\ quantities 
of the form $\pYX^*(y \cmid x)$ for every value of $x \in \X$.  The \emph{structural 
nested models} of \citet{robins:tsiatis:91} are an example of this.  These 
allow for interactions
between time-varying covariates and time-varying treatments, but they are always
marginal over future covariates; this makes them considerably more
flexible than marginal structural models, because they allow for dependence in
treatment decisions on all observed data.	 We again continue to make the necessary
assumptions for identifiability; see \citet{robins:tsiatis:91} for more detail. 

\begin{exm}[Structural Nested Models]
Suppose we have a sequence of binary treatments
$\Xt_1, \ldots, \Xt_T$ and time-varying covariates
$\Zc_1, \ldots, \Zc_T$, together with an outcome $Y$.
Let $\overline{\Zc}_t \equiv (\Zc_1, \ldots, \Zc_t)$
and $\underline{\Zc}_t \equiv (\Zc_t, \ldots, \Zc_T)$,
and similarly for $\overline{\Xt}_t$, $\underline{\Xt}_t$.
The \emph{structural nested model}
\citep{robins:tsiatis:91}
involves contrasts between $\xt_t = 0, 1$ of the form:
\begin{align*}
&p_{Y|\overline{\Zc}_t\overline{\Xt}_T}(y \cmid \overline{\zc}_t, \overline{\xt}_{t-1}; \Do(\xt_t, \underline{\xt}_{t+1} = 0)),
&& \forall \, \overline{\zc}_t, \overline{\xt}_{t-1}, t=0, \ldots, T. 
\end{align*}
The parameterization divides the effect of the treatments into
pieces corresponding to `blips' of effect at each time point: that is,
at each time $t$, we consider the effect of receiving treatment at
that time but no further treatment, versus never receiving any
treatment from time $t$ onwards.
The contrast may be
in the form of a risk difference, risk ratio or other suitable
quantity.  
%

We represent such a generic contrast by introducing a tilde above the variable being
contrasted; in the above example we would write:
\begin{align}
&p_{Y|\overline{\Zc}_t\overline{\Xt}_T}(y \cmid \overline{\zc}_t, \overline{\xt}_{t-1}; \Do(\widetilde{\xt}_t, \underline{\xt}_{t+1} = 0)),
&& \forall \, \overline{\zc}_t, \overline{\xt}_{t-1}, t=0, \ldots, T. \label{eqn:blip2}
\end{align}
See the more formal Definition \ref{dfn:basecont} below.
\end{exm}

We define two additional kinds of parameter to generalize these ideas.

\begin{dfn} \label{dfn:basecont}
Let $q_{\YIXZ}(y \cmid x, z)$ be a conditional distribution.
We denote by
$q_\YIXZ(y \cmid x^0, z)$ a \emph{baseline parameter}, which can
smoothly recover the relevant conditional distribution at
a particular baseline value $X=x^0$.

We will denote
by $q_\YIXZ(y \cmid \widetilde{x}, z)$ a \emph{contrast}
parameter (over $X$).
We define the pair of baseline and contrast parameters to be a
\emph{full} parameterization if, when we combine them,
we can smoothly recover all of $q_\YIXZ(y \cmid x, z)$.
\end{dfn}

In the appendix we give Lemma \ref{lem:basecont}, showing we can use risk differences,
risk ratios or odds ratios as contrast parameters, if $p > 0$ and each $X_t$ is 
binary.  
Examples of a set of baseline parameters might be $(\beta_0, \beta_z, \sigma^2)$
for some regression model $y = \beta_0 + \beta_x x + \beta_z z + \varepsilon$,
where $\Var \varepsilon = \sigma^2$; the natural contrast parameter would then be
$\beta_x$.  Alternatively it might  be the density $\pYIXZ(y \cmid x^0, z)$,
$y \in \Y, z \in Z$, for some value $x^0 \in \X$; the contrast parameter could
then be a risk ratio:
\begin{align*}
\pYIXZ(y \cmid \tilde{x}, z) \equiv \frac{\pYIXZ(y \cmid x, z)}{\pYIXZ(y \cmid x^0, z)} \quad \text{for all } x \in \X, y \in \Y, z \in \Z.
\end{align*}


\subsection{Iterated Frugal Parameterization}

How can we use the frugal parameterization to obtain the structural nested
model?  We now introduce the iterated frugal parameterization
to allow us to do just that.  

Consider a sequence of random variables $\Zc_1, \Xt_1, \Zc_2, \ldots, \Zc_T, \Xt_T$ 
and an outcome of interest $Y$.  Assume also that there is a natural 
`baseline' treatment level $\Xt_i = \xt_i^0$.  Then the \emph{iterated frugal 
parameterization} consists of a parameterization of `the past' (i.e.~$\pZX$), 
of $p^*_{Y|\Zc_1\overline{\Xt}_T}(y \cmid \zc_1, \overline{\xt}^0_T)$, 
and the following quantities:
\begin{align*}
\left.
\begin{array}{l}
p^*_{Y|\overline{\Zc}_{t}\overline{\Xt}_{T}}
(y \cmid \overline{\zc}_{t}, \overline{\xt}_{t-1}, \widetilde{\xt}_t, \underline{\xt}^0_{t+1})\\[6pt]
\phi_{Y\!\Zc_{t+1}|\overline{\Zc}_t\overline{\Xt}_t}^*
(y, \zc_{t+1} \cmid \overline{\zc}_{t}, \overline{\xt}_{t})
\end{array}
\right\}
\qquad
\begin{array}{l}
\forall y, \overline{\zc}_{T}, \overline{\xt}_{T}\\
t=1,\ldots,T,
\end{array}
\end{align*}
where the parameters can be used to obtain $p^*_{Y|\overline{\Zc}_{t}\overline{\Xt}_{T}}
(y \cmid \overline{\zc}_{t}, \overline{\xt}_{t-1}, \xt^0_t, \underline{\xt}^0_{t+1})$ such that
combined with $p^*_{Y|\overline{\Zc}_{t}\overline{\Xt}_{T}}
(y \cmid \overline{\zc}_{t}, \overline{\xt}_{t-1}, \widetilde{\xt}_t, \underline{\xt}^0_{t+1})$
we obtain a `full' parameterization (for $p^*_{Y|\overline{\Zc}_{t}\overline{\Xt}_{T}}
(y \cmid \overline{\zc}_{t}, \overline{\xt}_{t}, \underline{\xt}^0_{t+1})$). 
Note that if we consider the contrast parameter to be all the possible values 
other than the baseline value, then each
state of $\overline{\xt}_T$ will appear on the right-hand side of a
quantity $p^*_{Y|\overline{\Zc}_{t}\overline{\Xt}_{T}}$ exactly once.

\subsection{The Structural Nested Model}

How can we use a parameterization that incorporates
all the quantities (\ref{eqn:blip2})?  
Based on the temporal ordering, 
and given
$p_{Y|\overline{\Zc}_t\overline{\Xt}_T}(y \cmid \overline{\zc}_t, \overline{\xt}_{t-1}; \Do(\xt_t, \underline{\xt}_{t+1} = 0))$
and
\begin{align*}
p_{\Zc_{t+1}|\overline{\Zc}_t\overline{\Xt}_T}(\zc_{t+1} \cmid \overline{\zc}_t, \overline{\xt}_{t-1}; \Do(\xt_t, \underline{\xt}_{t+1} = 0)) = p_{\Zc_{t+1}|\overline{\Zc}_t\overline{\Xt}_t}(\zc_{t+1} \cmid \overline{\zc}_t, \overline{\xt}_{t}),
\end{align*}
we  need
$\phi^*_{Y\!\Zc_{t+1}|\overline{\Zc}_t\overline{\Xt}_t}(y, \zc_{t+1} \cmid \overline{\zc}_t, \overline{\xt}_{t})$
to recover the joint
$p_{Y\!\Zc_{t+1}|\overline{\Zc}_t\overline{\Xt}_T}(y, \zc_{t+1} \cmid \overline{\zc}_t, \overline{\xt}_{t}; \Do(\underline{\xt}_{t+1} = 0))$.
Then notice
\begin{align*}
p(y, \overline{\zc}_{T} \cmid \overline{\xt}_{t}; \Do(\underline{\xt}_{t+1})) = p(y, \overline{\zc}_{T} \cmid \overline{\xt}_{t+1}; \Do(\underline{\xt}_{t+2})) \cdot \frac{p(\xt_{t+1} \cmid \overline{\xt}_{t})}{p(\xt_{t+1} \cmid \overline{\zc}_t, \overline{\xt}_{t})},
\end{align*}
so we can `change worlds' and obtain probabilities with the same
settings from a reweighting that is identifiable from the previous
variables.  The following proposition gives the general result, proved and illustrated
by examples in Appendix \ref{sec:snmm_pf}.


\begin{prop} \label{prop:snmm}
We can parameterize $p_{\overline{\Zc}_T\overline{\Xt}_TY}(\overline{\zc}_T, \overline{\xt}_T, y)$ using smooth and regular parameterizations for 
$p^*_{Y|\Zc_1\overline{\Xt}_T}(y \cmid \zc_1, \overline{\xt}^0_T)$ and
\begin{align*}
\left.
\begin{array}{l}
p_{\Zc_{t}\hspace{-.5pt}\Xt_t|\overline{\Zc}_{t-1}\hspace{-.5pt}\overline{\Xt}_{t-1}}
(\zc_{t}, \xt_t \cmid \overline{\zc}_{t-1}, \overline{\xt}_{t-1}) \\[6pt]
p^*_{Y|\overline{\Zc}_{t}\overline{\Xt}_{T}}
(y \cmid \overline{\zc}_{t}, \overline{\xt}_{t-1}, \widetilde{\xt}_t, \underline{\xt}^0_{t+1})\\[6pt]
\phi_{Y\!\Zc_{t+1}|\overline{\Zc}_t\overline{\Xt}_t}^*
(y, \zc_{t+1} \cmid \overline{\zc}_{t}, \overline{\xt}_{t})
\end{array}
\right\}
\qquad
\begin{array}{l}
\forall y, \overline{\zc}_{T}, \overline{\xt}_{T}\\
t=1,\ldots,T,
\end{array}
\end{align*}
where each $p^*_{Y|\overline{\Zc}_t\overline{\Xt}_T}$ is cognate 
for the particular baseline $\underline{\xt}^0_{t+1}$. 
In particular, our parameterization can include `blips' such as those in (\ref{eqn:blip2}).  If either the contrast parameter is the odds ratio, or the
risk ratio and the outcome is positive and unbounded, then these pieces are also
variation independent.
\end{prop}

The proof for the special case of binary treatment variables is given in Appendix 
\ref{sec:snmm_pf}.  With this general formulation we do not require
$p^*_{Y|\overline{\Zc}_t\overline{\Xt}_T}$ to be of the
same form for each $t=1,\ldots,T$; this flexibility may be useful for many
settings.  However, we do need the baseline level $\underline{\xt}_{t+1}^0$
to be consistent over all $t$, since otherwise the inductive argument we use 
will not work.  Note also that
$\phi^*_{Y\!\Zc_{T+1}|\overline{\Zc}_T\overline{\Xt}_T}$ is trivial, since $\Zc_{T+1}$
is assumed constant.

Two numerical examples are given as R\ref{exm:snmm} and \ref{exm:snmm2} in 
Appendix \ref{sec:snmm_pf}.


\begin{rmk}
The \emph{History-Adjusted Marginal Structural Models} (HAMSMs)
introduced by \citet{vanderlaan:05} model (the mean of) the distributions
\[
p_{Y|\overline{\Zc}_t\overline{\Xt}_T}(y \cmid \overline{\xt}_{t-1}, \overline{\zc}_{t}; \Do(\underline{\xt}_{t})), \qquad t=1,\ldots,T.
\]
This is similar to the form of a structural nested mean model, but in
this case we attempt to model \emph{all} future treatment regimes simultaneously,
not just at a baseline $\underline{\xt}_{t} = 0$.  This effectively
requires us to model the association between $Y$ and each $\Xt_t$ multiple
times in different margins, and hence we will be using parameters that are
redundant; it therefore does not fall within our frugal framework.  This was pointed
out by \citet{robins:07}, who showed that it is a
non-congenial parameterization, and may lead to incompatible distributions.
\end{rmk}

%
%

\section{Discussion and Conclusion} \label{sec:con}

As we have demonstrated, the principle of a frugal parameterization is widely applicable and useful in many marginal modelling contexts, especially causal models.
We begin this discussion by briefly considering three more key settings for causal models: sensitivity analysis, instrumental variable (IV) analysis and mediation analysis.

In sensitivity analysis, a key challenge is to construct an \emph{augmented model} that
is compatible with the original model in the sense that it shares a marginal
distribution over the observed variables, but can be tweaked to introduce various levels
of unobserved confounding.  This is clearly possible within
our framework; considering Example R\ref{exm:havercroft_sim} in Appendix \ref{sec:vcop}, we can set 
the correlations involving $U$ to zero, and then increase them to test the dependence of our
conclusions to the presence of an unobserved confounder.

In an IV analysis, the instrument is used as an imperfect replacement for randomization when the actual treatment $X$ is affected by unobserved confounding. To formulate a generative IV model we typically want to combine a desired parameterization for $\pYX^*(y \cmid \Do(x))$ with a model that includes the IV and the confounder $U$.
The difficulty, here, is due to the particular properties of an IV which require the joint model to satisfy certain conditional independence properties while being compatible with the marginal causal model.
This is especially problematic for non-collapsible cases, for instance for logistic structural mean models \citep{robinsrotnitzkySMM,logisticSMM,clarkewind:2012} or structural Cox models \citep{coxIV}.  As outlined in Appendix \ref{sec:iv2}, we believe that our approach based on the frugal parameterization can also be helpful in these situations, but we leave details for future work.

In contrast, causal mediation analysis is an example where models contain
singularities and therefore our approach cannot be applied. Decomposing the
effect of a treatment $A$ on outcome $Y$ into the \emph{indirect effect}
via mediator $M$, and the remaining \emph{direct effect},  is  conceptually
the same as splitting $A$ into two
separate nodes $A,A'$, where observationally we always have $A=A'$; mediation
questions may then be considered as asking what would happen if $A \neq A'$
\citep{robins:10}.  The quantities of interest are therefore generally
functions of $p_{Y|AA'}(y \cmid \Do(a,a'))$, but where at the same time
$Y \indep A' \mid A, M$ holds in the full model where the two treatments
are potentially different \citep{didLIDA:19}.
Because this independence requires us to
model the $Y$-$A'$ association within the
joint distribution, not within the $(Y,A,A')$-margin, the only parameters that we are free to specify are then those of the
distribution of $Y$ given each level of $A$ (i.e.\ the strength of
the \emph{direct effect}); this is explicitly possible in the discrete case using results in \citet{evans:15}.
In other cases, attempts to specify both $p^*_{Y|AA'}$ and $p^*_{Y|AA'\!M}$ separately 
may lead to models which are not
compatible; for example, the equations (4) and (5) of \citet{loeys:13} do not
generally give a valid model because the logit function is not closed under marginalization.
\citet{lange:12} avoid the problem of explicitly modelling the joint distribution by
using marginal structural models instead, though their approach does not
allow for simulation from the resulting model.

Another example of nonsmoothness comes from quantities such as $\E_\theta[ Y \cmid \Do(x)] - \E_\theta[Y \cmid x]$, or some other contrast
between these two distributions.\footnote{This is related to (though distinct from) the
parameter used by \citet{hubbard:08} to estimate the effect of giving an entire
population a particular treatment, versus no intervention at all.}
This leads to a parameterization which is degenerate, in the sense that
its derivative (or nonparametric equivalent) is zero in some directions when
the two distributions are the same.

While such nonsmooth models still remain a challenge, we are certain that
marginal models based on a frugal parameterization have many further useful
applications and extensions worth exploring in future work.
For instance, classes of distribution that are
closed under marginalization and conditioning, such as MTP$_2$ distributions \citep{karlin80},  will naturally combine with our approach.
On the technical side, the proposed rejection sampling method can be inefficient,
and it would be
desirable to improve this by using more advanced methods,
along the lines of those suggested by \citet{jacob:20}.

As we noted in Section \ref{sec:existing},  we can see two opposing or
complementary trends in causal modelling:
many approaches are based on specifying structural causal models that implicitly
or explicitly condition on the entire past, and do not consider marginal objects
such as $\pYX(y \cmid \Do(x))$. 
In contrast, our approach is found in the books by \citet[][Chapter 3]{pearl:09},
\citet[][]{imbens:15} and \citet[][]{hernan:20}, which all consider marginal
causal quantities to be fundamental.
Beyond frugal parameterizations, we believe that thinking about
causal models as a form of marginal model, for which there is an older and
richer literature, may lead to many more advances in the field.

\subsection*{Acknowledgements}

We are grateful to Bohao Yao for some early simulations, as well as to
Thomas Richardson, James Robins, Ilya Shpitser, the Associate Editor and 
four anonymous reviewers for their insights and suggestions.  We would also
like to thank Qingyuan Zhao for reading a late draft and providing very insightful
comments and corrections, including the idea about a sensitivity analysis.  Part of
a revision of the manuscript was undertaken while both authors were Visiting 
Scientists at the Simons Institute, Berkeley.  

Section \ref{sec:data} was done as part of the IDEFICS 
Study\footnote{\url{http://www.idefics.eu}}. The data used
in this article cannot be shared publicly due to confidentiality policies agreed with the families
participating in the study. 
We gratefully acknowledge the financial 
support of the European Commission within the Sixth RTD Framework
Programme Contract No. 016181.  
The authors have no conflicts of interest to declare. 

\bibliographystyle{abbrvnat}
\bibliography{mybib}


\newpage

\appendix

%

\section{Smoothness, Regularity and Singularity} \label{sec:reg}

The first few definitions in this section are adapted from \citet{newey:90} 
and Chapter 5 of \citet{vandervaart:98}. 
Suppose that we have have a parametric family of distributions $\M = \{p_\theta : \theta \in \Theta \subseteq \mathbb{R}^d\}$, indexed by a parameter $\theta$.

\begin{dfn}
We say that the model $\M$ is \emph{differentiable in
quadratic mean} if there exists a function $\dot\ell(\theta_0)$ such that
as $\theta \to \theta_0$,
\begin{align*}
\int \left[\sqrt{p_\theta} - \sqrt{p_{\theta_0}} - \frac{1}{2}(\theta-\theta_0)^T \dot\ell(\theta_0) \sqrt{p_{\theta_0}} \right] \, d\mu = o(\|\theta - \theta_0\|^2).
\end{align*}
\end{dfn}

If a model is differentiable in quadratic mean we say that the parameterization
induced by $\theta$ is \emph{smooth}.
Now, for almost all statistical models of interest, $\dot\ell$ is
of course the \emph{score function}, that is
\begin{align*}
\dot\ell(\theta) = \frac{\partial}{\partial\theta} \log p_\theta.
\end{align*}
In this case, if the \emph{Fisher information} matrix
\begin{align*}
I(\theta) = \E \dot\ell(\theta)\dot\ell(\theta)^T
\end{align*}
is non-singular, then we also say that the map defined by
$\theta$ is a \emph{regular} parameterization.

We also have related but separate terminology for submodels, which we
adapt from \citet{drton:09}.

\begin{dfn}
Given a \emph{submodel} of $\M$, say $\M' \subseteq \M$, we say that
$\M'$ is \emph{nonsingular} if the induced subset of $\Theta$ is everywhere
locally Euclidean and of constant dimension.  Otherwise the model has
\emph{points of singularity} or \emph{singularities}.
\end{dfn}

An example of a model with singularities would be the union of
the axes $\{(\theta_1,\theta_2): \theta_1 \theta_2 = 0\}$, because this
model is not locally Euclidean at $\theta_1 = \theta_2 = 0$.

\section{Proof of Proposition \ref{prop:snmm}} \label{sec:snmm_pf}

We extend the notion of a risk difference, risk ratio or odds ratio 
to a general outcome variable (but still binary treatment) by writing
\begin{align*}
\operatorname{RD} &:= \pYIZX(y \cmid z, x=1) - \pYIZX(y \cmid z, x=0)\\
\operatorname{RR} &:= \frac{\pYIZX(y \cmid z, x=1)}{\pYIZX(y \cmid z, x=0)}\\
\operatorname{OR} &:= \frac{\pYIZX(y \cmid z, x=1) \cdot \pYIZX(y^0 \cmid z, x=0)}{\pYIZX(y^0 \cmid z, x=1) \cdot \pYIZX(y \cmid z, x=0)}
\end{align*}
for some arbitrary baseline value $y^0$. 
This latter definition is a special case of the one used by \citet{chen:07}.

\begin{lem} \label{lem:basecont}
Suppose $p > 0$, and that $X$ is binary.  Given $\pZX(z,x)$,
$\pYIZ(y \cmid z)$ and $\pYIZX(y \cmid z, \widetilde{x})$, where
$\widetilde{x}$ is contrasted using a risk difference, risk ratio or
an odds ratio, we can smoothly recover $\pYIZX(y \cmid z, x)$.  In 
addition, if we use the risk ratio and the range of $Y > 0$ is unbounded, 
or we use the odds ratio these three pieces will be variation independent.
\end{lem}

\begin{proof}
For a risk difference or ratio, it is clear that if $\pZX(z,x)$ and $\pYIZ(y \cmid z)$ are fixed, then
\begin{align*}
\theta_z &= \pYIZX(y \cmid z, x=1) - \pYIZX(y \cmid z, x=0)\\
\theta'_z &= \frac{\pYIZX(y \cmid z, x=1)}{\pYIZX(y \cmid z, x=0)}
\end{align*}
each give a regular representation of $\pYIZX(y \cmid z, x)$ when combined with
\begin{align*}
\pYIZ(y \cmid z) = \sum_{x=0}^1 p_\XIZ(x \cmid z) \pYIZX(y \cmid z, x).
\end{align*}
For the odds ratio we refer to \citet{chen:07} for details.  The variation 
independence of the odds ratio from its margins is well known 
\citep[e.g.][]{ruschendorf:95}.  If $Y$ is unbounded and 
$\pYIZX(y \cmid z, \widetilde{x})$ is the risk-ratio, then it is clear that we can 
modify it in any way and still obtain a valid joint distribution.  
\end{proof}

\begin{proof}[Proof of Proposition \ref{prop:snmm}] 
We consider the special case in which each $\Xt_t$ is binary, and 
proceed by induction on $T$. Note that we can combine all
the conditionals $p_{\Zc_{t}\Xt_t|\overline{\Xt}_{t-1}\overline{\Zc}_{t-1}}$
to obtain the joint distribution $p_{\overline{\Zc}_{T}\overline{\Xt}_T}$.
Now, by a simple adaptation of Theorem \ref{thm:main}, we start with
\begin{align*}
p^*_{Y|\Zc_1\overline{\Xt}_T}(y \cmid \zc_1, \underline{\xt}_1^0) \qquad p^*_{Y|\Zc_1\overline{\Xt}_T}(y \cmid \zc_1, \widetilde{\xt}_{1}, \underline{\xt}_2^0)
\qquad \phi^*_{Y\!\Zc_2|\Zc_1\Xt_1}(y, \zc_2 \cmid \zc_1, \xt_1),
\end{align*}
from which we can recover $p^*_{Y|\Zc_1\overline{\Xt}_T}(y \cmid \zc_1, \xt_1, \underline{\xt}_2^0)$ by Lemma \ref{lem:basecont}.
We can then combine with $p_{\Zc_2|\Zc_1\Xt_1}$ and $\phi^*_{Y\!\Zc_2|\Zc_1\Xt_1}$ to obtain
$p^*_{Y\!\Zc_2|\Zc_1\overline{\Xt}_T}(y,\zc_2 \cmid \zc_1, \xt_1, \underline{\xt}_2^0)$, and consequently (by reweighting)
$p_{Y\!\Zc_2|\Zc_1\overline{\Xt}_T}(y,\zc_2 \cmid \zc_1, \xt_1, \underline{\xt}_2^0)$.

Now, assume for induction that we can recover
$p_{Y|\overline{\Zc}_t\overline{\Xt}_T}(y \cmid \overline{\zc}_t, \overline{\xt}_{t-1}, \underline{\xt}_t^0)$; we have shown this for $t=2$.  We can reweight with some function of $p_{\overline{\Zc}_T \overline{\Xt}_T}$ to obtain $p^*_{Y|\overline{\Zc}_t\overline{\Xt}_T}(y \cmid \overline{\zc}_t, \overline{\xt}_{t-1}, \underline{\xt}_t^0)$,
and then combining with
$p^*_{Y|\overline{\Zc}_t\overline{\Xt}_T}(y \cmid \overline{\zc}_{t}, \overline{\xt}_{t-1}, \widetilde{\xt}_t, \underline{\xt}^0_{t+1})$
and again using Lemma \ref{lem:basecont} we obtain $p^*_{Y|\overline{\Zc}_t\overline{\Xt}_T}(y \cmid \overline{\zc}_t, \overline{\xt}_{t}, \underline{\xt}_{t+1}^0)$.
Then, we can
again use $\phi^*_{Y\!\Zc_{t+1}|\overline{\Zc}_t\overline{\Xt}_t}$ together with $p^*_{Y|\overline{\Zc}_{t}\overline{\Xt}_{T}}$ and $p^*_{\Zc_{t+1}|\overline{\Zc}_{t}\overline{\Xt}_{T}}$ (for $\underline{\Xt}_{t+1} = \underline{\xt}^0_{t+1}$)
to obtain $p^*_{Y\!\Zc_{t+1}|\overline{\Zc}_t\overline{\Xt}_t}$.  
Reweighting again yields an expression for
$p_{Y\!\Zc_{t+1}|\overline{\Zc}_t\overline{\Xt}_T}$ when $\underline{\Xt}_{t+1} = \underline{\xt}^0_{t+1}$, and hence $p_{Y|\overline{\Zc}_{t+1}\overline{\Xt}_T}(y \cmid \overline{\zc}_{t+1}, \overline{\xt}_{t}, \underline{\xt}_{t+1}^0)$.

Hence, by induction, we can obtain $p_{Y|\overline{\Zc}_T\overline{\Xt}_T}$, and consequently $p_{\overline{\Zc}_T\overline{\Xt}_TY}$.  

The results on variation independence follow directly from the implications in Lemma \ref{lem:basecont}.
\end{proof}

\begin{exmp} \label{exm:snmm}
Consider again the model in Figure \ref{fig:mod2}; in this case, we have
$\Xt_1 = A$ and $\Xt_2 = B$, with $\Zc_2 = L$ and $\Zc_1$ being null.
A structural nested mean model would include
\begin{align*}
&p_{\YIAB}(y \cmid \Do(a=b=0))
&&p_{\YIAB}(y \cmid \Do(a=1, b=0))\\
&\text{and }&&p_{\YIALB}(y \cmid a, \ell; \Do(\tilde{b})).
\end{align*}

In order to complete the parameterization we also need $p_{\ALB}$ and
$\phi^*_{Y\!L|A}$; the latter of these could be the conditional odds
ratio in the discrete case, for example.
%
%
%
The advantage of this representation of an SNM is that it makes
absolutely clear which (groups of) parameters are free to be varied.  Indeed,
like the previous examples this `model' is such that any
distribution over $A,L,B,Y$ (or more generally $\overline{\Xt}_T, \overline{\Zc}_T, Y$)
can be represented using this parameterization.

We demonstrate this by constructing a distribution for a structural
nested mean model over this graph.
We take all variables to be
binary, and let the blips be in the form of risk differences:
\begin{align*}
p_{\YIAB}(1 \cmid \Do(a=b=0)) &= 0.2\\
p_{\YIAB}(1 \cmid \Do(a=1, b=0)) - p_{\YIAB}(1 \cmid \Do(a=b=0)) &= 0.1\\
p_{\YIALB}(1 \cmid a, \ell; \Do(b=1)) - p_{\YIALB}(1 \cmid a, \ell; \Do(b=0)) &= 0.1 a + 0.05 \ell.
\end{align*}
Suppose also that $p_{A}(1) = 0.3$, and
\begin{align*}
p_{L|A}(1 \cmid a) &= 0.4 - 0.1 a\\
p_{B | AL}(1 \cmid a, \ell) &= 0.2 + 0.3a + 0.3\ell\\
\log \phi_{Y\!L|A}(1, 1 \cmid a) &= 0.1 + 0.1 a,
\end{align*}
where $\phi_{Y\!L|A}$ is the conditional odds ratio.
The resulting conditional probabilities
$p_{\YIALB}(1 \cmid a,\ell,b)$
are given in Table \ref{tab:snmm}.

\begin{table}
\centering
\begin{tabular}{ccc|c}
\toprule
$a$ & $\ell$ & $b$& $p_{\YIALB}(1 \cmid a,\ell,b)$\\
\midrule
0 & 0 & 0 & 0.194 \\
1 & 0 & 0 & 0.287 \\
0 & 1 & 0 & 0.210 \\
1 & 1 & 0 & 0.330 \\
0 & 0 & 1 & 0.194 \\
1 & 0 & 1 & 0.387 \\
0 & 1 & 1 & 0.260 \\
1 & 1 & 1 & 0.480 \\
\bottomrule
\end{tabular}
\caption{Table giving probability of survival from the SNMM in Example
R\ref{exm:snmm}.}
\label{tab:snmm}
\end{table}
\end{exmp}

\begin{exm} \label{exm:snmm2}

This is an expansion of Example R\ref{exm:snmm} in the notation of Section 7: 
hence $(A,B)$ becomes $(A_1, A_2)$, and $L$ becomes $L_2$.  We also add in 
a `static' covariate $L_1$ that is causally prior to all other variables.
Suppose that $T=2$, all variables are binary, and let the blips be in
the form of risk differences:
\begin{align*}
p_{Y|\Zc_1\overline{\Xt}_2}(1 \cmid \zc_1; \Do(\xt_1=\xt_2=0)) &= 0.2\\
p_{Y|\Zc_1\overline{\Xt}_2}(1 \cmid \zc_1; \Do(\xt_1=1, \xt_2=0)) - p_{Y|\Zc_1\overline{\Xt}_2}(1 \cmid \zc_1; \Do(\xt_1=\xt_2=0)) &= 0.1 + 0.1\zc_1\\
p_{Y|\overline{\Zc}_2\overline{\Xt}_2}(1 \cmid \zc_1, \xt_1, \zc_2; \Do(\xt_2=1)) - p_{Y|\overline{\Zc}_2\overline{\Xt}_2}(1 \cmid \zc_1, \xt_1, \zc_2; \Do(\xt_2=0)) &= 0.05 \zc_1 + 0.05 \zc_2 + 0.1 \xt_1.
\end{align*}
Suppose also that $p_{\Zc_1}(1) = 0.5$ and $p_{\Xt_1| \Zc_1}(1 \cmid \zc_1) = 0.3 + 0.3\zc_1$, and
\begin{align*}
p_{\Zc_2 | \Zc_1\Xt_1}(1 \cmid \zc_1, \xt_1) &= 0.4 + 0.3\zc_1 - 0.1\xt_1 - 0.2 \zc_1\xt_1\\
p_{\Xt_2 | \overline{\Zc}_2\Xt_1}(1 \cmid \zc_1, \zc_2, \xt_1) &= 0.2 + 0.3\xt_1 + 0.3\zc_2\\
\log \phi_{Y\!\Zc_2|\Zc_1\!\Xt_1}(1, 1 \cmid \zc_1, \xt_1) &= 0.1 + 0.1 \Ind_{\{\xt_1 = \zc_1\}},
\end{align*}
where $\phi_{Y\!\Zc_2|\Zc_1\!\Xt_1}$ is the conditional odds ratio.
The resulting conditional probabilities
$p_{Y|\overline{\Zc}_2 \overline{\Xt}_2}(1 \cmid \overline{\zc}_2, \overline{\xt}_2)$ are given in Table \ref{tab:snmm2}.

\begin{table}
\begin{center}
\begin{tabular}{cc|cc||c}
\toprule
$\zc_1$ & $\zc_2$ & $\xt_1$ & $\xt_2$ & $p_{Y|\overline{\Zc}_2 \overline{\Xt}_2}(1 \cmid \overline{\zc}_2, \overline{\xt}_2)$\\
\midrule
0 & 0 & 0 & 0 & 0.187 \\
  1 & 0 & 0 & 0 & 0.189 \\
  0 & 1 & 0 & 0 & 0.219 \\
  1 & 1 & 0 & 0 & 0.205 \\
  0 & 0 & 1 & 0 & 0.294 \\
  1 & 0 & 1 & 0 & 0.381 \\
  0 & 1 & 1 & 0 & 0.315 \\
  1 & 1 & 1 & 0 & 0.429 \\
\midrule
  0 & 0 & 0 & 1 & 0.187 \\
  1 & 0 & 0 & 1 & 0.239 \\
  0 & 1 & 0 & 1 & 0.269 \\
  1 & 1 & 0 & 1 & 0.305 \\
  0 & 0 & 1 & 1 & 0.394 \\
  1 & 0 & 1 & 1 & 0.531 \\
  0 & 1 & 1 & 1 & 0.465 \\
  1 & 1 & 1 & 1 & 0.629 \\
\bottomrule
\end{tabular}
\end{center}
\caption{Table giving probability of survival from the SNMM in Example \ref{exm:snmm2}.}
\label{tab:snmm2}
\end{table}
\end{exm}

\section{Vine Copulas} \label{sec:vcop}

As described in Example \ref{exm:cop}, a copula is a multivariate
CDF with uniform $(0,1)$ margins, and can be obtained from any
continuous parametric multivariate model by transforming each margin
using its univariate CDF.  However, there is a relative dearth of
multivariate families in dimensions greater than two, and this
limits the flexibility of such an approach.  One solution to this
problem has been to use \emph{vine copulas}, which chain together
bivariate families in order to give more flexible representations
of multivariate models.

We do not describe vine copulas in full generality here for the sake of
brevity, see \citet{bedford:02} for details.  Consider a system
of three variables, $U$, $L$ and $Y$.  In the case that $L \indep Y \mid U$,
we can model the joint distribution using two separate copulas, one each
for the $L,U$ margin and the $U,Y$ margin.  Due to the conditional
independence, the conditional quantiles of $L \cmid U$ and $Y \cmid U$
are uniformly distributed and uncorrelated.  It is then possible to
relax the conditional independence constraint, by placing another
copula model on these conditional quantiles.  Crucially, the distributions
of the original bivariate margins remain the same.

Vine copulas also have the nice property that for the second level
and above, parameters are conditional on the values of those at lower
levels; in particular they are variation independent.
As a comparison, the standard parameters of a jointly Gaussian
copula have to yield a positive definite matrix, which
is hard to enforce (other than by using the vine copula approach of
considering partial correlations).  This is particularly useful if
we introduce the treatment or other covariates as modifying the
parameters, since the link functions can be much simpler.

\begin{exmp} \label{exm:havercroft_sim}
We will again apply this to Example R\ref{exm:run} from \citet{havercroft12},
this time including the latent variable $U$.  We use Gaussian
copulas in a vine for the triple $(U,L,Y)$, with $U$-$L$ and $U$-$Y$ correlation parameters $2\expit(1) - 1 \approx 0.462$,
and $L$-$Y$ \emph{partial} correlation parameter $2\expit(0.5) - 1 \approx 0.245$.
We take $L$ and $Y$ to be exponentially distributed with means
\begin{align*}
\E [L\mid A=a] &= \exp(-(0.3 - 0.2a))\\
\E [Y\mid\Do(A=a, B=b)] &= \exp(-( - 0.5 + 0.2a + 0.3b)),
\end{align*}
as well as $A \sim \operatorname{Bernoulli}(\frac{1}{2})$ and $B \mid L=\ell, A=a \sim \operatorname{Bernoulli}(\expit(-0.3+0.4a + 0.3\ell))$; the marginal
distribution of $U$ plays no role, so we simply leave it as uniform.
We simulate a dataset of $n=10^4$ individuals, and again fitting via
IPW we obtain:
\begin{align*}
\hat\beta_0 &= -0.489 \; (0.022) & \hat\beta_a &= 0.202 \; (0.033) & \hat\beta_b &= 0.314 \; (0.029) & \hat\beta_{ab} &= -0.040 \; (0.042).
\end{align*}
Robust standard errors are shown in brackets, and each estimate is
indeed less than one standard error away from its respective nominal
value.  Code to replicate this analysis is contained in the vignette
\verb|Hidden_Variables| of the R package \texttt{causl}.
\end{exmp}

\section{Simulation Example} \label{sec:sim2}

We now apply the approach given in Section \ref{sec:sim_study} to
a single large dataset of size $n=10^4$.
%
Table \ref{tab:comp} shows the results, which this time are the estimates, standard
errors and bias.  We see that
our maximum likelihood method indeed has the jointly smallest standard errors, and 
that for each of the IPW, MLE, and doubly robust approaches 
the estimates are suggestive of consistency.  Only the outcome regression
model fails, and this is unsurprising since it is misspecified.
Code relating to this example is also found in the vignette \texttt{Comparison}
in the R package \texttt{causl}.

\begin{table}
\hspace{-1cm}
\begin{center}
%
\begin{tabular}[t]{lrrrrrrrrrrrr}
\toprule
\multicolumn{1}{c}{ } & \multicolumn{3}{c}{Outcome Regression} & \multicolumn{3}{c}{IP Weighting} & \multicolumn{3}{c}{Double Robust} & \multicolumn{3}{c}{MLE} \\
\cmidrule(l{3pt}r{3pt}){2-4} \cmidrule(l{3pt}r{3pt}){5-7} \cmidrule(l{3pt}r{3pt}){8-10} \cmidrule(l{3pt}r{3pt}){11-13}
  & Est. & SE & Bias & Est. & SE & Bias & Est. & SE & Bias & Est. & SE & Bias\\
\midrule
$\beta_0$ & $-$0.58 & 0.020 & $-$0.076 & $-$0.48 & 0.024 & 0.018 & $-$0.49 & 0.021 & 0.012 & $-$0.49 & 0.019 & 0.007\\
$\beta_a$ & 0.17 & 0.030 & $-$0.030 & 0.20 & 0.036 & $-$0.005 & 0.20 & 0.029 & $-$0.003 & 0.20 & 0.027 & $-$0.001\\
$\beta_b$ & 0.46 & 0.028 & 0.157 & 0.28 & 0.031 & $-$0.020 & 0.29 & 0.028 & $-$0.011 & 0.29 & 0.025 & $-$0.005\\
$\beta_{ab}$ & 0.04 & 0.040 & 0.042 & 0.03 & 0.045 & 0.026 & 0.02 & 0.053 & 0.024 & 0.02 & 0.034 & 0.019\\
\bottomrule
\end{tabular}
\caption{Table giving coefficients from the marginal structural
model via outcome regression (i.e.\ na\"ive regression on $A$ and $B$);
inverse probability weighting (IPW); doubly robust method (DR); and
our maximum likelihood approach (MLE).}
\label{tab:comp}
\end{center}
\end{table}

\section{Data Analysis} \label{sec:data2}

The analysis of \citeauthor{nohren:msc} consisted of using IPW with
a propensity score model based on the logistic regression model that
relates dichotomized fibre intake to
\begin{align*}
\text{country}\cdot\text{sex}\cdot\text{age}\cdot\text{age}^2 + \text{country}\cdot\text{isced} + \text{isced}\cdot \text{age} +
             \text{isced}\cdot\text{MVPA} + \text{vegscore}\cdot\text{AVM}
\end{align*}
as well the intercept and all other subsets of the terms above.  Here
isced is the average parental education level; AVM is the average time spent
with audiovisual media in hours per week; MVPA is the
average moderate-to-vigorous physical exercise performed in minutes per day;
vegscore is the vegetable score.
When we run the same analysis (indeed, the same code) for only the German
children, the results obtained are shown in Table \ref{tab:BIPS2}.



\begin{table}
\begin{center}
\begin{tabular}{cc|cc|cc}
\toprule
param. & coefficient & est. & s.e. & \multicolumn{2}{c}{95\% conf.\ int.} \\
\midrule
$\beta_1$ &  fibre  &  \phantom{$-$}0.331  &  0.247 & $-$0.153 & 0.814\\
$\beta_2$ &  PRS &    \phantom{$-$}0.497  &  0.208 & \phantom{$-$}0.089 & 0.906\\
$\beta_3$ &  PRS:fibre &        $-$0.492  &  0.452 & $-$1.377 & 0.393\\
\bottomrule
\end{tabular}
\caption{Table giving estimated coefficients in the marginal structural
model fitted by \citeauthor{nohren:msc} for effect modification of the
PRS on BMI by fibre intake, when applied to the same subset of the data
that we used.}
\label{tab:BIPS2}
\end{center}
\end{table}

\section{Young and Tchetgen Tchetgen Simulations} \label{sec:young}

The full model of \citeauthor{young:14} involves parameterizing
\begin{align*}
\frac{p_{Y_t|\overline{A}_tY_{t-1}}(1 \cmid \Do(\overline{a}_{t}), Y_{t-1} = 0)
}{p_{Y_t|\overline{A}_tY_{t-1}}(1 \cmid \Do(\overline{0}_{t}), Y_{t-1} = 0)} &= e^{\gamma(t, \overline{a}_t)} = \exp\left(\psi_0 a_t + \psi_1 a_{t-1} + \psi_{01} a_t a_{t-1}  \right). 
\end{align*}

We are also free to specify models for the dependence of each treatment
and the covariates upon previous treatments and covariates, as well as
the association parameters between each $Y_t$ and earlier covariates.  Again,
these can all be different for every $t$, but we follow \citeauthor{young:14}
who use logistic regressions for each variable.  They have
\begin{align*}
\logit p_{A_t|\overline{A}_{t-1}\overline{L}_t Y_{t-1}}(1 \cmid \overline{a}_{t-1}, \overline{\ell}_{t-1}, y_{t-1}=0)  &= \alpha_{*} + \alpha_0 \ell_t\\
\logit p_{L_t|\overline{A}_{t-1}\overline{L}_{t-1}Y_{t}}(1 \cmid \overline{a}_{t-1}, \overline{\ell}_{t-1}, y_{t}=0)  &= \beta_1 a_{t-1}.
\end{align*}
They also use a logistic regression for the distribution of survival given
the treatments and covariates, but we want to parameterize directly in terms
of the $\psi$s.
We therefore define
\begin{align*}
\logit p_{Y_t|\overline{A}_t\overline{L}_tY_{t-1}}(1 \cmid \overline{a}_t, \overline{\ell}_t, y_{t-1}=0) &= \theta_* + \theta_{a0} a_{t} + \theta_{\ell0} \ell_{t} + \theta_{a1} a_{t-1},
\end{align*}
noting that the parameters $\theta_{a0}$ and $\theta_{a1}$ are not actually free, 
because they are a function of the other parameters after specifying 
$\psi_0$, $\psi_1$ and $\psi_{01}$.

\citeauthor{young:14} specify the vectors $\alpha = (0.5,0.5)$, $\beta_1 = -2$ and $\theta=(-7,-0.5,-0.8,0)$ and then use
the g-formula (\ref{eqn:gform_yt}) to compute
the corresponding values of $\psi_0,\psi_1,\psi_{01}$.  We will specify the values
of $\psi$ as well as $\theta_*$ and $\theta_{\ell0}$, and then compute the new values of
other elements of $\theta$.  Note that all of the values of $\psi_0$ used are
very close to $-0.8$, which is a consequence of the rare outcome assumption
made by the original authors. 


Continuing the example from Section \ref{sec:survival},
 we simulate datasets of size $n=10^5$
and a variety of values for $\beta_1$ and $\theta_{a0}$, with $\theta_{\ell0} = -0.8$.

\begin{table}
\centering
\begin{tabular}{r|rr|rr|rr}
\toprule
 \multicolumn{1}{c|}{$\beta_1$} &  \multicolumn{1}{c}{$\theta_{a0}$} &  \multicolumn{1}{c|}{$\operatorname{Bias}(\hat\theta_{a0})$} &  \multicolumn{1}{c}{$\theta_{\ell0}$} &  \multicolumn{1}{c|}{$\operatorname{Bias}(\hat\theta_{\ell0})$} & \multicolumn{1}{c}{$\psi_0$} &  \multicolumn{1}{c}{$\operatorname{Bias}(\tilde\psi_0)$} \\
\midrule
 $-2.0$ & $-2.0$ & $-0.0005$ & $-0.8$ & $0.0023$ & $-0.79955$ & $-0.0079$ \\
 $-0.5$ & $-0.5$ & $0.0004$ & $-0.8$ & $-0.0017$ & $-0.79957$ & $0.0024$ \\
 $0.0$ & $-0.5$ & $-0.0035$ & $-0.8$ & $0.0005$ & $-0.79957$ & $-0.0024$ \\
 $-0.5$ & $0.0$ & $0.0018$ & $-0.8$ & $-0.0003$ & $-0.79950$ & $0.0009$ \\
 $0.5$ & $-2.0$ & $-0.0305$ & $-0.8$ & $0.0022$ & $-0.79955$ & $-0.0041$ \\
 $2.0$ & $-2.0$ & $-0.0574$ & $-0.8$ & $0.0008$ & $-0.79955$ & $-0.0029$ \\
\bottomrule
\end{tabular}
\caption{Table showing bias in estimates from the survival model of \citet{young:14}.
The values given for each parameter are the precise values chosen,
and $\hat\theta$ is the MLE, while $\tilde{\psi}$ is estimated via inverse
probability weighting.
The sample bias in these estimates' mean is shown in the adjacent
column; we performed $N=5\,000$ runs with sample size $n=10^5$.}
\label{tab:young1}
\end{table}

Table \ref{tab:young1} shows the bias that results in maximum likelihood
estimates of $\theta_{a0}$ and estimates of $\psi_0$ via inverse
probability weighting (compare this with Table I of \citealp{young:14}).  We
can see that this is indeed still small, implying that our simulation method works as
expected.

\section{Instrumental Variables} \label{sec:iv2}


One common causal approach, when faced with unobserved confounding, is to use an \emph{instrumental variables} (IV) model, as shown in Figure \ref{fig:iv}.  In this case interest may be in the average causal effect which is a function of the quantity  $\pYX(y \cmid \Do(x))$; other popular IV approaches consider causal estimands such as the `complier causal effect' or the `effect of treatment on the treated' which we do not further address, here. The average causal effect, if everything is linear, can be identified by the ratio $\Cov(Z,Y)/\Cov(Z,X)$.
More challenging is the case where the effect of $X$ on $Y$ is non-linear.

\begin{figure}
\begin{center}
  \begin{tikzpicture}
  [rv/.style={circle, draw, thick, minimum size=6mm, inner sep=0.75mm}, node distance=20mm, >=stealth]
  \pgfsetarrows{latex-latex};
\begin{scope}
  \node[rv]  (1)              {$Z$};
  \node[rv, right of=1] (2) {$X$};
  \node[rv, right of=2, xshift=-10mm, yshift=14mm, color=red] (U) {$U$};
  \node[rv, right of=2] (3) {$Y$};

  \draw[->, very thick, color=blue] (1) -- (2);
  \draw[->, very thick, color=blue] (2) -- (3);
  \draw[->, very thick, color=red] (U) -- (2);
  \draw[->, very thick, color=red] (U) -- (3);
  \end{scope}
  \end{tikzpicture}
 \caption{A representation of the instrumental variables model.}
  \label{fig:iv}
  \end{center}
\end{figure}
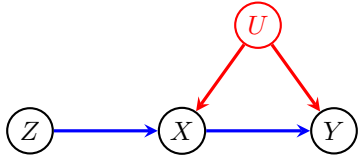

We can use our framework to simulate from the general IV model, by explicitly including the hidden
variable $U$.  We first parameterize the distribution of the `past', i.e.\ $(U,Z,X)$
so that $U \indep Z$;
then we take the distributions $\pYX(y \cmid \Do(x))$ and the association
parameter $\phi^*_{Y,U\!Z|X} = \phi^*_{Y,U|X}$ so as not to depend upon $Z$ at all. This
will allow us to simulate from an IV model, provided that the pieces $p_{U\!Z\!X}$,
$\pYX^*$ and $\phi^*_{YU|X}$ are
chosen from a sufficiently rich family of distributions.

Specifically, suppose that we want to simulate from a particular model from
Figure \ref{fig:iv}, with a specified parametric form for $\pYX^*(y \cmid x)$
(presumably this is $\pYX(y \cmid \Do(x))$).
Then we should use the following algorithm:
\begin{enumerate}
\item select a model $\theta^*_{Y|X}$ for $\pYX^*(y \cmid x)$;
\item choose a distribution for $(U,Z,X)$ such that
$U$ and $Z$ are independent;
\item choose a model for $\phi^*_{Y,U|X} = \phi^*_{Y,U\!Z|X}$
(i.e.\ such that $Y \indep Z \mid X, U$).
\end{enumerate}

Now, combine these to obtain the resulting joint distribution.  In
particular note that even if $Y$ is binary, we can simulate using a 
copula model and then
dichotomize $Y$ from the resulting continuous distribution.  This
works particularly well with a probit or logistic model, for example.

This gives a basic outline of how to represent an instrumental variable model so that we can simulate exactly from (almost\footnote{Since
it must satisfy A\ref{ass:tails}.}) any model of this kind.
To reiterate Section \ref{sec:alg}, we simulate by sampling
from $p^*_{U\!ZY|X}$, and then rejecting samples based on the value
of $p_{X|U\!Z}/p^*_{X|U\!Z}$.
However, further work is needed to extend this to structural mean models for IV analyses. These build on a particular no-effect modification assumption within a marginal (over unobserved confounders) model that is conditional on the natural treatment value and the IV, a restriction which cannot always be represented in a structural equation type model \citep{robinsrotnitzkySMM, clarke10identification}.


%

\end{document}